\documentclass[journal]{IEEEtran}
\usepackage{dsfont}
\usepackage{xcolor}
\usepackage[pdftex]{graphicx}
	\graphicspath{{./figures/}}
\usepackage{subfigure}
\usepackage{enumitem}

\usepackage{setspace}
\usepackage[colorlinks=true,linkcolor=blue,citecolor=red,hyperfootnotes,pdftex]{hyperref}
\usepackage{cite}

\usepackage{latexsym}
\usepackage{amssymb}
\usepackage[cmex10]{amsmath}
\usepackage{amsfonts}
\usepackage[mathscr]{eucal}
\usepackage{amsthm}

\usepackage{array}
\usepackage{multirow}

\usepackage{algorithmicx}
\usepackage{algorithm}
\usepackage{algpseudocode}

\usepackage{graphicx}
	\graphicspath{{./figures/}}
\usepackage{float}
\usepackage{url}
\usepackage{eurosym}
\usepackage{setspace}
\usepackage{balance}
\usepackage{dsfont}
\usepackage{bbm}
\usepackage{array}
\usepackage{caption}

\newcommand{\mybackslash}[2]{%
	\raisebox{-2pt}[9pt][5pt]{{#1}\hspace*{-4pt}%
	\rotatebox{30}{\Big\backslash}\hspace*{-4pt}%
	\raisebox{3pt}[0pt][0pt]{#2}}}
\usepackage{multirow}

\def\ontop#1#2{\setbox0\hbox{#2}\copy0\llap{\raise\ht0\hbox{#1}}}

\DeclareFontFamily{U}{mathx}{\hyphenchar\font45}
\DeclareFontShape{U}{mathx}{m}{n}{<5> <6> <7> <8> <9> <10> <10.95> <12> <14.4> <17.28> <20.74> <24.88> mathx10}{}
\DeclareSymbolFont{mathx}{U}{mathx}{m}{n}
\DeclareMathSymbol{\bigtimes}{1}{mathx}{"91}

\DeclareMathAlphabet{\mts}{U}{rsfs}{m}{n}

\theoremstyle{definition}
\newtheorem{Def}{Definition}

\theoremstyle{plain}
\newtheorem{Thm}{Theorem}
\newtheorem{Cor}{Corollary}
\newtheorem{Prp}{Proposition}
\newtheorem{Lem}{Lemma}

\theoremstyle{remark}
\newtheorem{Rem}{Remark}

\newcommand{\zbx}[1]{\raisebox{0pt}[0pt][0pt]{#1}}                     % adjust height for inline math
\renewcommand{\,}{\hspace*{+1pt}}                                      % horizontal spacing command
\usepackage{multirow}

\definecolor{gr}{rgb}{0.50,0.50,0.50}
\definecolor{wh}{rgb}{1.00,1.00,1.00}

\algnewcommand{\Initial}              {\item[{\bf initialization:}]}
\algnewcommand{\Initiax}{\item[\textcolor{wh}{\bf initialization:}]}

\algrenewcommand{\algorithmicrequire}       {{\bf input:}}
\algnewcommand{\Requirx}{\item[\textcolor{wh}{\bf input:}]}

\algrenewcommand{\algorithmicensure}        {{\bf output:}}
\algnewcommand{\Outputx}{\item[\textcolor{wh}{\bf output:}]}

\algrenewcommand{\algorithmiccomment}[1]{\hfill\textcolor{gr}{$\triangleright$\ \zbx{#1}}}

\algrenewtext{If}[1]    {{\bf\algorithmicif\ }{#1}{\bf\ \algorithmicthen}}
\algrenewtext{ElsIf}[1] {{\bf\algorithmicelse\ \algorithmicif\ }{#1}{\bf\ \algorithmicthen}}
\algrenewtext{Else}     {{\bf\algorithmicelse}}
\algrenewtext{EndIf}    {{\bf\algorithmicend}}

\algrenewtext{For}[1]   {{\bf\algorithmicfor\ }{#1}{\bf\ \algorithmicdo}}
\algrenewtext{ForAll}[1]{{\bf\algorithmicfor\ all\ }{#1}{\bf\ \algorithmicdo}}
\algrenewtext{EndFor}   {{\bf\algorithmicend}}

\algrenewtext{While}[1] {{\bf\algorithmicwhile\ }{#1}{\bf\ \algorithmicdo}}
\algrenewtext{EndWhile} {{\bf\algorithmicend}}

\algrenewtext{Repeat}   {{\bf\algorithmicrepeat\ }}
\algrenewtext{Until}[1] {{\bf\algorithmicuntil\ }{#1}}

\algrenewcommand{\algorithmicindent}{8pt}

\ifCLASSINFOpdf
\else
\fi
\usepackage{amsmath}

\allowdisplaybreaks
\begin{document}
\title{Dynamic Information Sharing\\
and Punishment Strategies}

\author{Konstantinos Ntemos, 
        George Pikramenos, and Nicholas Kalouptsidis
\thanks{ 
All authors are with the Dept.
of Informatics and Telecom., National and Kapodistrian University of Athens, Athens, Greece, e-mail: \{kdemos, gpik, kalou\}@di.uoa.gr.}}%

% \markboth{Journal of \LaTeX\ Class Files,~Vol.~14, No.~8, August~2015}%
% {Shell \MakeLowercase{\textit{et al.}}: Bare Demo of IEEEtran.cls for IEEE Journals}



\maketitle

\begin{abstract}
In this paper we study the problem of information sharing among rational self-interested agents as a dynamic game of asymmetric information. We assume that the agents imperfectly observe a Markov chain and they are called to decide whether they will share their noisy observations or not at each time instant. We utilize the notion of {\em conditional mutual information} to evaluate the information being shared among the agents. The challenges that arise due to the inter-dependence of agents' information structure and decision-making are exhibited. For the finite horizon game we prove that agents do not have incentive to share information. In contrast, we show that cooperation can be sustained in the infinite horizon case by devising appropriate punishment strategies which are defined over the agents' {\em beliefs} on the system state. We show that these strategies are closed under the best-response mapping and that cooperation can be the optimal choice in some subsets of the state belief simplex. We characterize these {\em equilibrium regions}, prove uniqueness of a {\em maximal} equilibrium region and devise an algorithm for its approximate computation.
\end{abstract}


% \fi
\IEEEpeerreviewmaketitle

\section{Introduction}
\IEEEPARstart{T}{he} process of information sharing is important in a wide range of applications of 
high socio-economic impact, including distributed estimation and detection \cite{Yu_2015}, cyber-security \cite{Miehling_2018}, social networking \cite{Jiang_2014} and viral marketing \cite{Chen_2010}. In such cases 
autonomous agents with enhanced {\em decision-making} capabilities, disseminate information, in a dynamic fashion, according to individual motives. It then turns out that the processes of information sharing and decision-making 
are {\em interdependent}; the decisions of an agent affect the information structure of their peers, which in turn affect their {\em optimal} decision-making. Thus, there is a need for a joint study of information sharing and decision making. This need is addressed in this work.

We consider two self-interested agents who seek to track the state of a Markov Chain. The agents are equipped with sensing capabilities that enable them to obtain noisy observations about the underlying state. Moreover, agents are offered the possibility to share their measurements with other agents; sharing information may enhance estimation performance, while the decision to share information entails some transmission cost. Estimation quality can be assessed by a variety of popular performance measures. We focus on an information utility function that measures the reduction in uncertainty. The transmission cost is assumed exogenous and constant over time. The difference between the expected estimation benefit offered by information sharing and the transmission cost defines the instantaneous reward of each agent. 
Each agent takes into account not only current payoffs, but also expected future rewards accumulated over time. We refer to the above setup as a {\em Dynamic Information Sharing Game} (DISG).

From a technical perspective, DISG is a {\em dynamic game of asymmetric information}, since agents have access to different information sets which are unknown to their peers. These games are notoriously hard to deal with, because agents have to reason about the private information of others by forming {\em beliefs} \cite{Tirole, Nayyar_2014, Vasal_2019}. These beliefs are {\em interdependent} with the agents' strategies, making the computation of optimal behavior challenging. In this complex setting, the question of whether information sharing can be sustained at equilibrium is of relevant importance. We show that this is only possible in the infinite horizon setting and propose a class of punishment strategies that can form information sharing equilibria.

\subsection{Related work}
Information sharing has been studied in several research areas. In \cite{Chakravorty_2016}, the problem of interactive communication between users that obtain noisy measurements about a state variable was investigated. The users are allowed to exchange information in the form of quantized symbols. 
The problem was modelled as a {\em team problem} and a dynamic programming algorithm was derived for the computation of the optimal strategies. In \cite{Yu_2015}, \cite{Ntemos_2018}, strategic information sharing was studied in the context of wireless networks with the agents being interested in a parameter estimation task. The authors utilized the bounded rationality assumption \cite{Shoham_2008} in order to describe the agents’ decision-making process.  Bounded rationality can cast the model into a realistic setting that is applicable to real people’s behavior in some cases. In this paper, we study the information sharing process assuming fully rational agents.

Strategic information sharing, with focus on the design of economic incentives to stimulate cooperation among agents, has been studied in \cite{Heegard_2016,Gal-or_2005,Naghizadeh_2017,Laube_2017}. More specifically, the economic incentives of information exchange for multi-operator service delivery were analyzed in \cite{Heegard_2016}. A game theoretic model was used to show that sharing of information can be sustained at equilibrium given that there is mutual, long-term cooperation among operators. In the context of cyber-security, \cite{Gal-or_2005} studied the incentives of competitive firms to share security information through a third-party authority and the impact on social welfare. In the same context, \cite{Naghizadeh_2017} studied strategic information sharing among firms. The authors modelled firms' interactions as an $N$-agents Prisoner's Dilemma and designed incentives for sustainable cooperation. An excellent survey on strategic information sharing in cyber-security is provided in \cite{Laube_2017}. 

In this work, agents are interested in an estimation task and the expected instantaneous rewards depend on the information agents acquire. Hence, the expected {\em value of information} is tied to the estimation problem and can not be treated as an exogenous variable. In turn, the information the agents possess is also endogenous, since it depends on agents' decisions. The consideration of these coupled dynamics in the context of fully rational agents differentiates our work from the above studies.

Our study entails three key features: $(i)$ decision-making under partial observability of the state, $(ii)$ asymmetric information structure and $(iii)$ punishment strategies. 

The study of the single-agent dynamic {\em optimal decision-making} has a long history tracing back to the seminal works on stochastic control \cite{Bellman_1966}, \cite{Blackwell_1964}. If the state evolves as a Markov chain but is partially observed by the agent, strategies are formed as functions of available information and assessed in terms of accumulated expected rewards. The properties of the optimal course of action are studied within the framework of Partially Observed Markov Decision Processes (POMDPs) \cite{Astrom_1965}, \cite{Sondik_1973}. It turns out that there is no loss in optimality if strategies are functions over beliefs on the current state, i.e. over probability distributions of the current state given the available information. Optimal strategies can be computed by several exact and approximate algorithms \cite{Sondik_1973}, \cite{Cassandra_thesis}.

In the strategic information sharing setup we consider, the POMDP model needs to be extended to capture the fact that two agents are present, who observe the Markov source by proprietary sensors and decide whether to share their data to increase their own rewards. The instantaneous reward of each agent depends on the willingness of the other agent to share information. Indeed sharing enhances in general the quality of estimation since a richer set of observations is available. This dependence introduces a coupling in the agents’ behaviors. Moreover, agents generally have access to different information sets, giving rise to the second feature of DISG, namely asymmetry of information. This asymmetry necessitates a departure from classical information structures which assume all past observations and actions are known to all agents. Thus, dynamic games of asymmetric information constitute a natural framework for the study of optimum information sharing strategies.

In dynamic games of asymmetric information the agents need to form beliefs about other agents' private information, along with the computation of their optimal strategies. Beliefs and strategies are inter-dependent and sequential decomposition is in general not possible. The study of stochastic games of asymmetric information is an active research area with significant recent developments \cite{Nayyar_2014,Gupta_2014,Ouyang_2015, Ouyang_2017, Tavafoghi_2016,  Vasal_2016b, Vasal_2019,Vasal_2016}. The notion of {\em Common Information-Based Markov Perfect Equilibria} (CIB-MPE) was proposed in \cite{Nayyar_2014} to capture beliefs over states and on private information of all agents and to use these common beliefs as drivers for policy choice. Under the assumption that the belief update mechanism is {\em strategy independent}, a backward induction sequential procedure was developed for the calculation of CIB-MPE in the finite horizon case. The same techniques were applied in \cite{Gupta_2014} to the linear-Gaussian case.

More general cases with {\em strategy dependent} common beliefs were explored under the presence of {\em signaling} \cite{Ouyang_2015, Ouyang_2017, Tavafoghi_2016,  Vasal_2016b, Vasal_2019}. Signaling occurs when agents reveal part of their private information through their strategies. In \cite{Ouyang_2015, Ouyang_2017} the authors introduce a subclass of PBEs, namely {\em Common Information Based Perfect Bayesian Equilibria} (CIB-PBE), prove the existence of CIB-PBEs for a subclass of such games and develop a dynamic programming sequential decomposition to compute them. Dynamic games of asymmetric information with delayed information structure and hidden actions are investigated in \cite{Tavafoghi_2016}. Signaling equilibria are investigated in the context of linear quadratic Gaussian games in \cite{Vasal_2016b}. In \cite{Vasal_2019} the authors introduced a subclass of PBE, called {\em Structured Perfect Bayesian Equilibria} (SPBE) and described a two-step backward-forward recursive algorithm to find SPBEs. The backward algorithm solves a fixed point equation on the space of probability simplices and defines an equilibrium generating function. This function is utilized to define equilibrium strategies and beliefs through a forward recursion. These techniques were also applied to the context of Bayesian learning and the study of informational cascades \cite{Vasal_2016}.

Our work comes as a complement to the above references through the study of information structures that are not characterized by a predefined protocol, but are directly affected by the agents' actions. Furthermore, the agents are interested in estimating an underlying state giving rise to a utility function that is a non-linear function of the belief and captures the expected gain offered by the information exchange. A third differentiating factor is the consideration of novel punishment strategies, inspired by the literature on {\em repeated games} 
\cite{Fudenberg_1986,Mailath_2006}. This class of strategies 
is not a subject of study of the above works.
 
Punishment strategies have been studied in connection with {\em folk theorems}. Folk theorems have been extended to stochastic games with complete state information \cite{Dutta_1995}, \cite{Horner_2011} and private types \cite{Escobar_2013}, \cite{Sugaya_2012} without recourse to POMDP models that are critically employed in this work. We point out that folk theorems, which are asymptotic results and deal with the issue of whether any feasible individually rational payoff can be attained for sufficiently high discount factor, are beyond the scope of our work.
\subsection{Our work and contributions}
In this paper we develop a general model of strategic information sharing, where
 two agents aim to track a Markov chain based on observations on the state. The proposed approach  quantifies the {\em value of received information} through the concept of {\em conditional mutual information} \cite{Gallager}, although more general reward functions can be used without affecting the validity of results%, or the net reduction in entropy of the underlying system state brought forward by the data exchange mechanism
. The agents decide to share their measurements on the basis of discounted rewards that tradeoff expected estimation gains and transmission costs. We use the concepts and methodologies of the works on dynamic games with non-classical information structures to show that agents’ beliefs are strategy-dependent and to demonstrate that the finite horizon setting rules out information sharing at equilibrium.

In contrast, we prove that sustainable cooperation can emerge in the infinite-horizon case. This is done by introducing a form of punishment strategies, inspired by grim-trigger, which we call Constrained Grim Trigger (CGT) strategies. CGT strategies are parametrized by subsets of the belief simplex. We prove that CGT strategies are closed under the best-response mapping, which means that an agent can respond against a CGT strategy with a CGT strategy without loss of optimality. We show that under such strategies, cooperation can be sustained in some subsets of the belief state simplex, which we call {\em equilibrium regions}. We prove the uniqueness of a maximal equilibrium region and devise a fixed-point-like algorithm for its approximate computation. Finally, results that ensure nonemptiness of the maximal cooperation region are given.

The above results are illustrated experimentally through simulations where the POMCP algorithm \cite{Silver_2010} is used to visualize the equilibrium region. 
The findings of this work could find applications in settings where endogenizing the decision to share information is meaningful. Potential applications include Bayesian learning and the study of informational cascades and distributed networks with adversarial agents.
\subsection{Notation}
Random variables are denoted by upper case letters; their realizations by the corresponding lower case letters. For $a<b$, %and $c<d$ 
the notation $X_{a:b}$ denotes the vector $(X_a,\ldots,X_b)$. For a statement $s$, $\mathds{1}_{\{s\}}=1$ if $s$ is true, while $\mathds{1}_{\{s\}}=0$ if $s$ is false. The $\bot$ symbol is used to denote contradiction and the $\bigotimes$ symbol is used to denote the Cartesian product.
%\IEEEpubidadjcol

%\subsubsection{Subsubsection Heading Here}
%----------------------------------------------------------%
%                      Section Change                      %
%----------------------------------------------------------%
\section{Dynamic Information Sharing}
\label{system_model-ch-5}
The ingredients of the basic model are presented in this section. We consider two agents seeking to track the state $X_t$ of a Markov chain. For this purpose each agent has access to measurements obtained by private sensors. The agents have the option of sharing information. The decision to share observations assesses the trade off between transmission costs and estimation gains brought by the additional measurements. These statements are made precise below.
\par{\noindent{\bf State dynamics and observation models}}. $X_t$ takes values in a finite set $\mathcal{X}$. Each agent $n\in\mathcal{N}=\{1,2\}$ receives observation $Y^n_t$ at time $t=0,1,2,\ldots$. The random variables $Y^n_t$ take values in the finite sets $\mathcal{Y}^n$. At each time $t$ agents decide simultaneously whether they will share observations or not. Thus, the set of possible actions for both agents is $\mathcal{A}=\{0,1\}$. Let $A^n_t$ denote the action of agent $n$ at time $t$. $A^n_t=1$ means that agent $n$ sends her private observations $Y^n_t$ to the other agent (denoted by $–n$), whereas $A^n_t=0$ means that agent $n$ sends no data. The state evolves exogenously and is not affected by agents’ actions and observations. More precisely it holds
\begin{equation}
\label{state_equation}
\mathbb{P}(X_{t+1}|X_{0:t},Y^n_{0:t},Y^{-n}_{0:t},A^n_{0:t},A^{-n}_{0:t}) =\mathbb{P}(X_{t+1}|X_t).
\end{equation}
Observations are conditionally independent given the current state and are governed by the model
\begin{align}
\label{obs_equation}
&\mathbb{P}(Y^n_t,Y^{-n}_t\vert X_{0:t},Y^n_{0:t-1},Y^{-n}_{0:t-1},A^n_{0:t-1},A^{-n}_{0:t-1})\nonumber\\
&= \mathbb{P}(Y^n_t|X_t)\mathbb{P}(Y^{-n}_t|X_t).
\end{align}
\par{\noindent{\bf Data exchange}}. At each time $t$ agent $n$ (resp. $-n$) receives the signal $Z^{-n}_t$ (resp. $Z^n_t$) which is a deterministic function of the agent's observation $Y^{-n}_t$ (resp. $Y^n_t$) and the action of agent $-n$ (resp. $n$). Here we shall assume that either the observation is shared error free, or no relevant data is shared. Thus, the data exchange mechanism is described by
\begin{equation}
\label{zeta_function}
Z^n_t(Y^n_t,A^n_t)=
\begin{cases}
Y^n_t, &\text{if} \ \ \ A^n_t=1,\\
\epsilon, &\text{if} \ \ \ A^n_t=0.
\end{cases}
\end{equation}
$\epsilon$ signifies that no information is shared. Clearly, $Z^n_t\in\tilde{\mathcal{Y}}^n_t$, where $\tilde{\mathcal{Y}}^n_t=\mathcal{Y}^n_t\cup\{\epsilon\}$.
\par{\noindent{\bf Information sets}}. The information available to agent $n$ at time $t$, $I^n_t$ is formed by the private history $I^{n,p}_t$ and the common history $I^c_t$.
\begin{align}
\label{agent_information}
I^n_t=(I^{n,p}_t,I^c_t).
\end{align}
The common history is known to both agents and consists of the agents' actions (i.e., $A^{1:2}_{1:t-1}$) and the history of the exchanged signals (i.e., $Z^{1:2}_{1:t-1}$), while the private history $I^{n,p}_t$ is known only to agent $n$ and includes all the observations that agent $n$ decided not to share until the present time $t$. These histories at the beginning of time $t$ are defined as follows
\begin{align}
\label{agent_information1}
&I^c_t=(Z^1_{0:t-1},Z^2_{0:t-1},A^1_{0:t-1},A^2_{0:t-1}),\\
\label{agent_information2}
&I^{n,p}_t
=(Y^n_k\vert A^n_k=0, \ 0\leq k<t).
\end{align}
Let $\mathcal{I}^n_t,\mathcal{I}^{n,p}_t,\mathcal{I}^c_t$ be the sets of all possible agent's {\em $n$ histories}, agent $n$'s {\em private histories} and {\em common histories} at time $t$, respectively.  Initially, at time $t=0$ the common information is $I^c_0=\tilde{\pi}_0$, where $\tilde{\pi}_0$ is the {\em common prior} belief on state $X_0$, and 
evolves as
\begin{equation}
\label{Common_history_evolution_eq}
I^c_{t+1}=
\begin{cases}
(I^c_t,A^1_t,A^2_t,Z^1_t=\epsilon,Z^2_t=\epsilon), &\text{if} \, \, \ A^1_t=A^2_t=0,\\
(I^c_t,A^1_t,A^2_t,Y^1_t,Z^2_t=\epsilon), &\text{if} \, \, \ A^1_t=1,A^2_t=0,\\
(I^c_t,A^1_t,A^2_t,Z^1_t=\epsilon,Y^2_t), &\text{if} \, \, \ A^1_t=0,A^2_t=1,\\
(I^c_t,A^1_t,A^2_t,Y^1_t,Y^2_t), &\text{if} \, \, \ A^1_t=A^2_t=1.
\end{cases}
\end{equation}
The private information of agent $n$ at time $t=0$ is $I^{n,p}_0=\emptyset$ for all $n$ 
and it is updated as
\begin{equation}
\label{Private_history_evolution_eq}
I^{n,p}_{t+1}=
\begin{cases}
I^{n,p}_t, &\text{if} \, \, \ A^n_t=1,\\
(
I^{n,p}_t,Y^n_t), &\text{if} \, \, \ A^n_t=0.
\end{cases}
\end{equation}
If agent $n$ decides $A^n_t=1$, then $Y^n_t$ is added in the common information $I^c_{t+1}$, otherwise it is added in  $I^{n,p}_{t+1}$. Note that the two sets $I^c_t$ and $I^{n,p}_t$ never overlap and an observation that belongs to one set does not belong to the other.
\par\noindent{\bf Agents' strategies}. Let $g=(g^1,g^2)$ be a {\em strategy profile} consisting of both agents' strategies. Agent $n$'s strategy $g^n=(g^n_1,\ldots ,g^n_T)$ in {\em finite horizon} or $g^n=(g^n_1,\ldots ,g^n_\infty)$ in {\em infinite horizon} is a collection of control laws $g^n_t$ which map agent $n$'s available information at time $t$ to a probability distribution over the agent's actions ({\em behavioral strategies}) i.e., $g^n_t:\mathcal{I}^n_t\rightarrow \Delta(\mathcal{A})$, $n\in\{1,2\}$, where
\begin{align}
\label{strategy_definition}
\mathbb{P}^{g^n_t}(A^n_t=a^n_t|I^n_t=i^n_t)=g^n_t(i^n_t)(a^n_t).
\end{align}
The set of all possible behavioral strategies of agent $n$ at time $t$ is denoted as $\mathcal{G}^n_t$.

The timing of events at a time $t$ is illustrated in Fig. \ref{Timing} and is as follows:
\begin{enumerate}
\item The state is $X_t$ and the agents' histories are $I^n_t=(I^{n,p}_t,I^c_t)$, $n\in\{1,2\}$.
\item Both agents select their actions $A^1_t,A^2_t$.
\item Both agents send signals $Z^1_t,Z^2_t$, according to the selected actions $A^1_t,A^2_t$ ({\em see \eqref{zeta_function}}).
\item Common and private histories are updated according to \eqref{Common_history_evolution_eq} and \eqref{Private_history_evolution_eq}, respectively.
\end{enumerate}
\begin{figure}[t]
\centering
\includegraphics[width=0.35\textwidth,height=0.28\textwidth]
{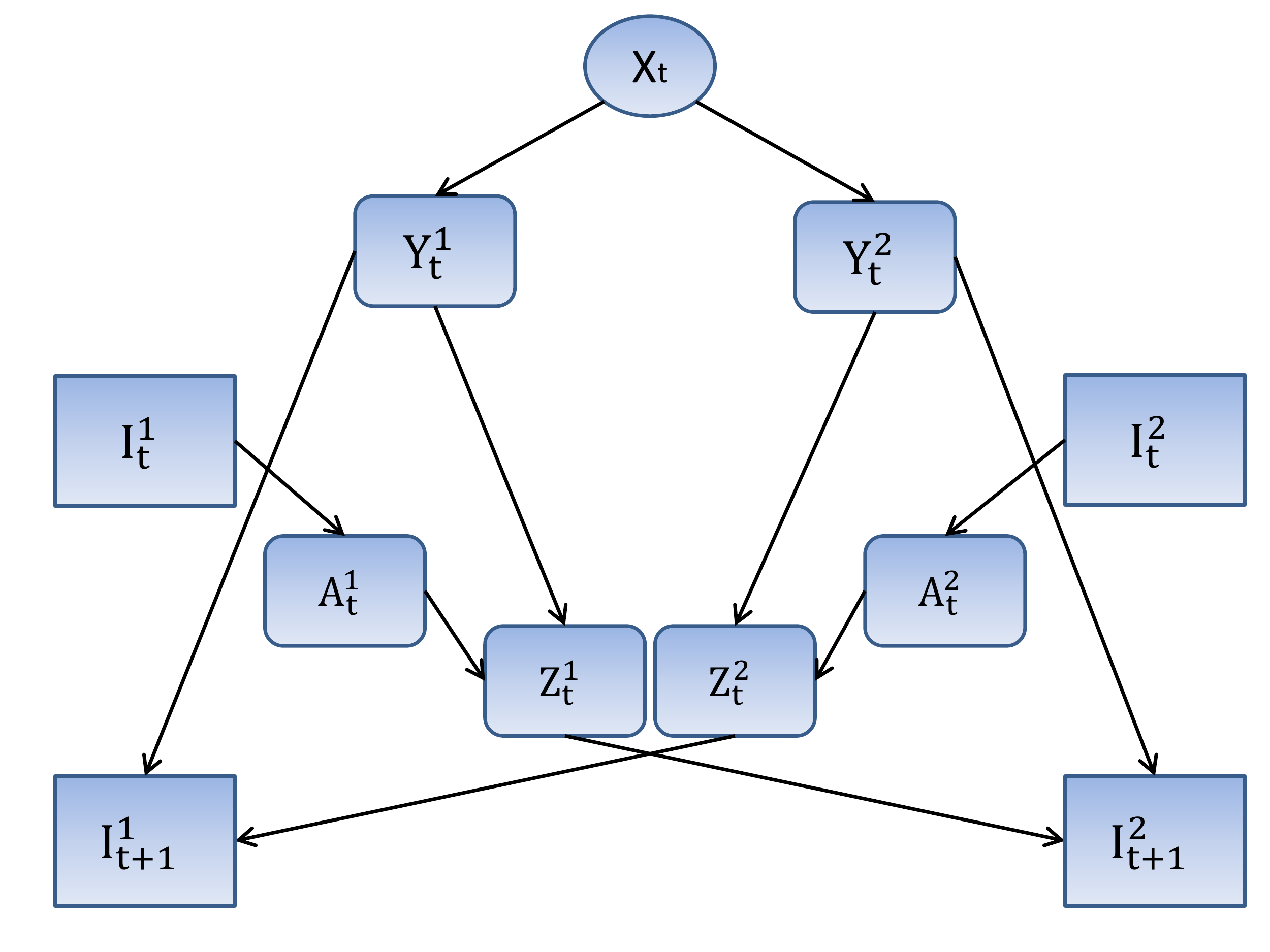}
\caption{Timing of events and dependencies of variables (top to bottom).}
\label{Timing}
\end{figure}
\section{Rewards, beliefs and equilibria}%and value of information}
\label{sec_rewards}
The potential benefits of {\em information sharing} are captured by a utility function that balances the instantaneous {\em estimation performance} and transmission cost. Agent $n$ incurs a transmission cost $c^n>0$ when it sends information (i.e., $a^n_t=1$), and $0$ when no information is sent (i.e., $a^n_t=0$). Thus, the transmission cost is $a^n_tc^n$. The reception gain can be quantified by several performance metrics. Here, we consider the odds of improving the estimate of the state probability upon receiving the signal $z^{-n}_t$ against the estimate of the state probability computed without the shared information. %by using private information only. 
Using logs, and for given realizations $x_t,y^n_t,z^{-n}_t,i^n_t,a^n_t$, the corresponding likelihood ratio for agent $n$ is
\begin{align}
\label{instantaneous_reward}
&r^n_t(x_t,y^n_t,z^{-n}_t, i^n_t)=
\log\frac{\mathbb{P}(x_t\vert z^{-n}_t,y^n_t,i^n_t)}{\mathbb{P}(x_t\vert y^n_t,i^n_t)}.
\end{align}
Clearly, $r^n_t(\cdot)$ depends on the history realization $i^n_t$. 
To save notation, we drop $i^n_t$ when it is clear from the context. 
Since the state $X_t$, the agent's observation $Y^n_t$, as well as the received information from the other agent $Z^{-n}_t$ is unknown at time $t$, agent $n$ needs to take expectation on \eqref{instantaneous_reward} given the information it possesses at that time, $i^n_t$. Thus, the {\em expected instantaneous reception gain} becomes
\begin{align}
\label{expected_instantaneous_reward}
&\mathbb{E}\{r^n_t(X_t,Y^n_t,Z^{-n}_t
)\vert i^n_t\}=
I(X_t;Z^{-n}_t\vert Y^n_t,i^n_t)\nonumber\\
&=H(X_t|Y^n_t,i^n_t)-H(X_t|Y^n_t,Z^{-n}_t,i^n_t),
\end{align}
where $I(X_t;Z^{-n}_t|Y^n_t,i^n_t)$ is the {\em conditional mutual information} of $X_t$ and $Z^{-n}_t$ given $Y^n_t,i^n_t$. $H(X_t|Y^n_t,i^n_t)$ and $H(X_t|Y^n_t,Z^{-n}_t,i^n_t)$ denote the {\em conditional entropy} of $X_t$ given $Y^n_t,i^n_t$ and the conditional entropy of $X_t$ given $Y^n_t,Z^{-n}_t,i^n_t$, respectively \cite{Gallager}.
\begin{Rem}
\label{information_functions}
Eq. \eqref{expected_instantaneous_reward} shows explicitly the contribution of the other agent (through $Z^{-n}_t$) in the reduction of
uncertainty about state $X_t$. The mutual information belongs to the class of information utility functions \cite{Naghshvar_2012,Kartik_2019,Coleman_2009}. Information utilities have been successfully employed in applications such as active sequential hypothesis testing \cite{Naghshvar_2012}, \cite{Kartik_2019} and codes for communication channels with feedback \cite{Coleman_2009}. In general, information utility functions employ a suitable measure of uncertainty and model the reduction of uncertainty at each stage. Besides the reduction in entropy employed in this work, several other related uncertainty measures have been used such as the {\em extrinsic Jensen-Shannon divergence} \cite{Naghshvar_2012}, the {\em average confidence level} \cite{Kartik_2019} and the expected reduction in the KL distance \cite{Coleman_2009}.

The mutual information between the channel input and channel output has been used as payoff function in the study of communication in the presence of jamming as a zero sum game \cite{McEliece_1983,Borden_1985,Stark_1988}
. In these games the encoder tries to maximize the mutual information, while the jammer tries to minimize it by introducing noise in the channel. Power allocation games using the 
mutual information have been extensively studied in %multi user channels, 
MIMO communications \cite{Palomar_2003}. Mutual information has been utilized in machine learning \cite{Gupta_2008} as a metric of performance and in neurosciences \cite{Friston_2017} as well. The use of more general reward functions is discussed in subsection \ref{discussion_rewards}.
\end{Rem}
Based on the above, the expected instantaneous reward for agent $n$ and a specific action $a^n_t$ becomes
\begin{align}
\label{Instantaneous_reward}
    &\mathbb{E}\{R^n_t(X_t,Y^n_t,Z^{-n}_t,a^n_t)\vert I^n_t=i^n_t\}\nonumber\\
    &=
\mathbb{E}\{r^n_t(X_t,Y^n_t,Z^{-n}_t)\vert i^n_t\}-a^n_tc^n\nonumber\\
&=I(X_t;Z^{-n}_t\vert Y^n_t,i^n_t)-a^n_tc^n\nonumber\\
&=
H(X_t|Y^n_t,i^n_t)-H(X_t|Y^n_t,Z^{-n}_t,i^n_t)-a^n_tc^n_t.
\end{align}
Evaluation of \eqref{Instantaneous_reward} requires the computation of $\mathbb{P}(x_t,y^n_t,z^{-n}_t\vert i^n_t)$ and its marginals. Unless specific conditions are imposed, this computation involves a complex intertwining of sharing decisions and beliefs on both the unknown state and the private information of the other agent. This is clarified in the sequel.
\vspace{-0.15cm}
\subsection{Expected instantaneous reward}
The expected instantaneous reception gain \eqref{expected_instantaneous_reward} yields
\begin{align}
\label{ins_reward}
&H(X_t|Y^n_t,I^n_t=i^n_t)-H(X_t|Y^n_t,Z^{-n}_t%(Y^{-n}_t,A^{-n}_t)
,I^n_t=i^n_t)\nonumber\\
&=\sum_{y^n_t\in\mathcal{Y}^n}\mathbb{P}(y^n_t|I^n_t=i^n_t)H(X_t|y^n_t,I^n_t=i^n_t)\nonumber\\
&-\sum_{y^n_t\in\mathcal{Y}^n,z^{-n}_t\in\tilde{\mathcal{Y}}^n}\mathbb{P}(y^n_t,z^{-n}_t|I^n_t=i^n_t)H(X_t|y^n_t,z^{-n}_t,I^n_t=i^n_t)\nonumber\\
&=-\sum_{y^n_t,x_t}\mathbb{P}(y^n_t|x_t)\mathbb{P}(x_t|i^n_t)\log\frac{\mathbb{P}(y^n_t|x_t)\mathbb{P}(x_t|i^n_t)}{\sum_{x_t}\mathbb{P}(y^n_t|x_t)\mathbb{P}(x_t|i^n_t)}\nonumber\\
&+\sum_{y^n_t,z^{-n}_t,x_t}\mathbb{P}(y^n_t|x_t)\mathbb{P}(z^{-n}_t|x_t,i^n_t)\mathbb{P}(x_t|i^n_t)\nonumber\\
&\times\log\frac{\mathbb{P}(y^n_t|x_t)\mathbb{P}(z^{-n}_t|x_t,i^n_t)\mathbb{P}(x_t|i^n_t)}{\sum_{x_t}\mathbb{P}(y^n_t|x_t)\mathbb{P}(z^{-n}_t|x_t,i^n_t)\mathbb{P}(x_t|i^n_t)}.
\end{align}
In the above expression, the terms $\mathbb{P}(x_t|i^n_t)$ and $\mathbb{P}(z^{-n}_t|x_t,i^n_t)$ need to be further discussed. $z^{-n}_t$ is a function of $y^{-n}_t$ and $a^{-n}_t$ ({\em see} \eqref{zeta_function}). Note that
\begin{align}
\label{ins_reward2}
&\mathbb{P}(z^{-n}_t=\epsilon|x_t,i^n_t)=\mathbb{P}(a^{-n}_t=0|x_t,i^n_t),\\
\label{ins_reward3}
&\mathbb{P}(z^{-n}_t=y^{-n}_t|x_t,i^n_t)=\mathbb{P}(a^{-n}_t=1|x_t,i^n_t)
\mathbb{P}(y^{-n}_t|x_t),
\end{align}
for any $y^{-n}_t\in\mathcal{Y}^{-n}$. Moreover, the distribution of $a^{-n}_t$ is given by agent $-n$'s strategy $g^{-n}_t(i^{-n}_t)(a^{-n}_t)$. Even if $g^{-n}_t$ is known, agent $n$ needs to reason about the private history of agent $-n$, as $g^{-n}_t $ is a function of $i^{-n}_t=(i^c_t,i^{-n,p}_t)$ ({\em see} \eqref{strategy_definition}). Hence, marginalization over agent $-n$'s private information $i^{n,p}_t$ yields
\begin{align}
\label{infer_action}
&\mathbb{P}(a^{-n}_t=a|x_t,i^n_t)=\sum_{i^{-n,p}_t}\mathbb{P}(a^{-n}_t=a|x_t,i^{-n,p}_t,i^{n,p}_t,i^c_t)\nonumber\\
&\times\mathbb{P}(i^{-n,p}_t|i^n_t,x_t)=\sum_{i^{-n,p}_t}g^{-n}_t(i^{-n}_t)(a^{-n}_t=a)\mathbb{P}(i^{-n,p}_t|i^n_t,x_t)\nonumber\\
&=\sum_{i^{-n,p}_t}g^{-n}_t(i^{-n}_t)(a)\frac{\mathbb{P}(x_t|i^{-n,p}_t,i^n_t)\mathbb{P}(i^{-n,p}_t|i^n_t)}{\mathbb{P}(x_t|i^n_t)}.
\end{align}
Thus, in order to calculate the expected instantaneous reward for a given $g^{-n}$, agent $n$ needs to form a {\em belief} about the state as well as agent $-n$'s private information.

In the following Lemma we identify cases where the computation of \eqref{expected_instantaneous_reward} given the other agent's strategy, does not require inference on the other agent's private information and we provide a simpler formula for computing the expected instantaneous reception gain in such cases.

Proofs are relegated to the Appendix.
\begin{Lem}
\label{Lem_reward}
The expected instantaneous reception gain \eqref{expected_instantaneous_reward} is given by
\begin{align}
\label{MI_rem}
&\mathbb{E}\{r^n_t({X}_t,Y^n_t,Z^{-n}_t)\vert I^n_t=i^n_t\}%=I(X_t;Z^{-n}_t\vert Y^n_t,i^n_t)
\nonumber\\
&=g^{-n}_t(i^c_t)(a^{-n}_t=1)I(X_t;Y^{-n}_t|Y^n_t,i^n_t),
\end{align}
if either of the following is true:
\begin{enumerate}
\item both agents have access to the same information, i.e., $i^1_t=i^2_t=i^c_t$.
\item $g^{-n}_t(i^{-n,p}_t,i^c_t)(a^{-n}_t)=g^{-n}_t(\bar{i}^{-n,p}_t,i^c_t)(a^{-n}_t)$ for every $a^{-n}_t$ and for every $i^{-n,p}_t\neq\bar{i}^{-n,p}_t$ such that $i^{-n,p}_t,\bar{i}^{-n,p}_t\in\mathcal{I}^{-n,p}_t$.
\end{enumerate}
\end{Lem}
Under statement $2)$, strategies of agent $-n$ remain invariant for all possible realizations of private information $I^{-n,p}_t$. The intuition is that the event $Z^{-n}_t=\epsilon$ does not contribute to the average reward under the stated assumptions; the mutual information between $X_t$ and $Z^{-n}_t$ is that between $X_t$ and $Y^{-n}$ provided that the sharing action $a^{-n}=1$ is chosen.
\subsection{Perfect Bayesian Equilibrium}
Agents' total expected rewards starting from a time $t$ in the finite horizon case are given by
\begin{align}
\label{Finite_horizon_total_reward}
\mathbb{E}\Big \{\sum^T_{j=t}R^n_j(X_j,Y^n_j,Z^{-n}_j,A^n_j)\vert i^n_t
\Big\},
\end{align}
Discounted total expected rewards in the infinite horizon case are given by
\begin{align}
\label{Infinite_reward}
\mathbb{E}\Big\{\sum^\infty_{j=t}\delta^j R^n_j(X_j,Y^n_j,Z^{-n}_j,A^n_j)\vert i^n_t
\Big\},
\end{align}
where 
$\delta\in[0,1)$ is a discount factor, which is common for both agents. The expectation is w.r.t. all random variables, including states, observations, and actions.

The problem formulated above constitutes a dynamic game of asymmetric information. An appropriate solution concept is Perfect Bayesian Equilibrium (PBE) \cite{Tirole}. A PBE is a generalization of Subgame Perfect Equilibrium (SPE) for asymmetric information games that considers a consistent belief system on other agents' private information so as to verify the sequential rationality of the strategies.

From the history of the game some part is known to agent $n$ and another part is unknown. The unknown part consists of the system states and the observations that the other agent has decided not to share. Each agent assesses the total expected rewards of a strategy profile \eqref{Finite_horizon_total_reward}, \eqref{Infinite_reward} by forming beliefs about the unknown parts in the history of the game. The collection of beliefs over the whole time horizon is called {\em belief profile} and is denoted as $\mu=(\mu^1,\mu^2)$, where $\mu^n=\{\mu^n_t\}_{t\in\mathcal{T}}$. For the finite horizon case, it is $\mathcal{T}=\{0,\dots,T\}$, while for the infinite horizon case $\mathcal{T}=\mathbb{N}$. 
$\mu^n_t$ is defined as
\begin{align}
\label{belief_general}
\mu^n_t(i^n_t)(X_{0:t},I^{-n,p}_t)=\mathbb{P}^{g^n,g^{-n}}(X_{0:t},I^{-n,p}_t|I^n_t=i^n_t).
\end{align}
A PBE is an {\em assessment} i.e., a pair of strategy and belief profiles $(g^*,\mu)$ that requires {\em sequential rationality} of the strategies and {\em consistency} of beliefs. An assessment $(g^*,\mu)$ is sequentially rational if $\forall t\in \mathcal{T}, i^n_t\in\mathcal{I}^n_t, n\in\{1,2\}$, $g^{n,*}_{t:T}$ is a solution to
\begin{align}
\label{sequential_rationality}
\sup_{g^n_{t:T}\in\mathcal{G}^n_{t:T}}\mathbb{E}^{g^n_{t:T},g^{-n,*}_{t:T}}_{\mu^n_t}\{\sum^T_{j=t}R^n_j(X_j,Y^n_j,Z^{-n}_j%(Y^{-n}_t,A^{-n}_t)
,A^n_j)|i^n_t\}.
\end{align}
The definition of sequential rationality is similar in the infinite horizon case.

Adapting the definition given in \cite{Ouyang-thesis,Tavafoghi_thesis}, we call an assessment $(g^*,\mu)$  consistent if $\forall t\in \mathcal{T}$ and $n\in\{1,2\}$, if $i^n_{t+1}$ and $i^n_t$ are such that $\mathbb{P}^{g^*}_{\mu}(i^n_{t+1}|i^n_t)>0$, $\mu^n_{t+1}(i^n_{t+1})$ must satisfy Bayes' rule. On the other hand, if $i^n_{t+1}$ and $i^n_t$ are such that $\mathbb{P}^{g^*}_{\mu}(i^n_{t+1}|i^n_t)=0$, then
\begin{align}
\label{belief_4}
&\mu^n_{t+1}(i^n_{t+1})(x_{0:t+1},i^{-n,p}_{t+1})>0,
\end{align}
only if
\begin{align}
\label{belief_5}
\hat{\mu}^n_{t+1}(i^n_{t+1})(x_{0:t+1},i^{-n,p}_{t+1})>0.
\end{align}
where
\begin{align}
\label{sf-beliefs}
\hat{\mu}^n_t(i^n_t)(x_{0:t},i^{-n,p}_t)=\mathbb{P}^{(a^1_{0:t-1},a^2_{0:t-1})}(x_{0:t},i^{-n,p}_t\vert i^n_t).
\end{align}
The so called {\em signaling-free} belief system \cite{Ouyang-thesis} $\hat{\mu}=(\hat{\mu}^1,\hat{\mu}^2)$, where $\hat{\mu}^n=\{\hat{\mu}^n_t\}_{t\in\mathcal{T}}$, employs a sequence of actions generated in an {\em open-loop} fashion. This way it is ensured that the beliefs off the equilibrium path are consistent with the system dynamics and observations models.

In the context of the DISG model, the Bayes' rule governing consistency for on equilibrium path beliefs (i.e., $i^n_{t+1}$ and $i^n_t$ are such that $\mathbb{P}^{g^*}_{\mu}(i^n_{t+1}|i^n_t)>0$) takes the following form
\begin{align}
\label{belief_finite2}
&\mu^n_{t+1}(i^n_{t+1})(x_{0:t+1},i^{-n,p}_{t+1})=\frac{\mathbb{P}^{g^*}_{\mu}(i^n_{t+1},x_{0:t+1},i^{-n,p}_{t+1}|i^n_t)}{\mathbb{P}^{g^*}_{\mu}(i^n_{t+1}|i^n_t)}\nonumber\\
&=g^{-n}_t(i^{-n,p}_t,i^c_t)(a^{-n}_t)\mathbb{P}(x_{t+1}|x_t)\mathbb{P}(y^n_t|x_t)\mathbb{P}(y^{-n}_t|x_t)\nonumber\\
&\times
\frac{\mu^n_t(i^n_t)(x_{0:t},i^{-n,p}_t)
}{W^n_t},
\end{align}
where $W^n_t$ is given by
\begin{align}
\label{belief_aux}
    & W^n_t=%\nonumber\\
    \sum_{i^{-n,p}_t,x_{0:t}} (g^{-n}_t(i^{-n,p}_t,i^c_t)(a^{-n}_t)\mathbb{P}(z^{-n}_t|x_t,a^{-n}_t)\mathbb{P}(y^n_t|x_t)\nonumber\\
    &\times\mu^n_t(i^n_t)(x_{0:t},i^{-n,p}_t)),
\end{align}
and $\mathbb{P}(z^{-n}_t|x_t,a^{-n}_t)$ is given by
\begin{align}
\label{y_z}
&\mathbb{P}(z^{-n}_t|x_t,a^{-n}_t)\\
&=\mathds{1}_{\{z^{-n}_t=\epsilon,a^{-n}_t=0\}}+\sum_{y^{-n}_t}\mathds{1}_{\{z^{-n}_t=y^{-n}_t,a^{-n}_t=1\}}\mathbb{P}(y^{-n}_t|x_t).\nonumber
\end{align}
\noindent
Eqs. \eqref{belief_finite2}-\eqref{y_z} follow by utilizing \eqref{agent_information}, \eqref{Common_history_evolution_eq}, \eqref{Private_history_evolution_eq}, \eqref{strategy_definition} and by distinguishing between cases $a^{-n}_t=0$ and $a^{-n}_t=1$.

Eqs. \eqref{belief_finite2}, \eqref{belief_aux} and \eqref{y_z} and a simple induction argument demonstrate that agent $n$'s belief $\mu^n_t$ does not depend on her own strategy $g^n_{0:t-1}$, but in general %may 
depends on the other agent's strategy given that $I^{-n,p}_t\neq\emptyset$, meaning
\begin{align}
\label{Lemma2_statement}
&\mu^n_t(i^n_t)(X_{0:t},I^{-n,p}_t)=\mathbb{P}^{g^n,g^{-n}}(X_{0:t},I^{-n,p}_t|I^n_t=i^n_t)\nonumber\\
&=\mathbb{P}^{g^{-n}}(X_{0:t},I^{-n,p}_t|I^n_t=i^n_t).
\end{align}
\subsection{Finite horizon}
In the finite horizon case, information sharing can never occur under a PBE equilibrium. This is stated in the following theorem.
\begin{Thm}
\label{Optimal_strategy_finite_horizon}
In the finite horizon DISG, the set of PBEs is fully characterized by 
$$\mathcal{B}=\{(g^{NC},\mu)\ \vert\ \mu \text{ is a consistent belief profile w.r.t. } g^{NC}\},$$ where $g^{n,NC}=\{g^{n,NC}_t\}_{t\in \mathcal{T}}$ with $g^{n,NC}_t(i^n_t)(a^n_t=1)=0$ $\forall n,t,i^n_t$. Equivalently, an assessment $(g^*,\mu)$, with $\mu$ consistent with respect to (w.r.t.) $g^*$, is a PBE of the finite horizon DISG if and only if $g^*\equiv g^{NC}$.
\end{Thm}
We proved that $g^{NC}$ is optimal for any consistent belief system.  Under $g^{NC}$ the agents' private histories at every time $t$ are comprised of all their past observations, meaning
\begin{align}
\label{non-cooperation-private-history}
&I^{n,p}_t=(Y^n_{1:t-1}), \quad\forall t, n.
\end{align}
So, under the $g^{NC}$ strategy profile agent $n$ needs to assign at each time $t$ a probability distribution from every realization $i^n_t$ over the part of history that is unknown to $n$ (i.e., $X_{0:t},I^{-n,p}_t$) 
and the belief defined in \eqref{belief_general} must be {\em consistent}. Starting from initial belief $\mu^n_0(x_0)=\tilde{\pi}_0(x_0)$, 
\eqref{belief_finite2} yields
\begin{align}
\label{belief_3}
&\mu^n_{t+1}(i^n_{t+1})(x_{0:t+1},i^{-n,p}_{t+1})=\mathbb{P}^{g^{NC}}_\mu(x_{0:t+1},i^{-n,p}_{t+1}|i^n_{t+1})=\nonumber\\
&=\mathbb{P}^{g^{NC}}_\mu(x_{t+1},x_{0:t},y^{-n}_t,i^{-n,p}_t|i^n_t,y^n_t,a^n_t=0,a^{-n}_t=0)\nonumber\\
&=\frac{\mathbb{P}(x_{t+1}|x_t)\prod_{j\in\{1,2\}}\mathbb{P}(y^j_t|x_t)\mu^n_t(i^n_t)(x_{0:t},i^{-n,p}_t)}{\sum_{x_{0:t},i^{-n,p}_t}\mathbb{P}(y^n_t|x_t)\mu^n_t(i^n_t)(x_{0:t},i^{-n,p}_t)}.
\end{align}
For the {\em off-equilibrium paths} the signaling-free belief system \eqref{sf-beliefs} can be used.

Theorem \ref{Optimal_strategy_finite_horizon} is 
in accordance with the intuition behind the result of no sustainable cooperation in finite-horizon repeated \textit{Prisoner's dilemma} \cite{Mailath_2006}. The proof of no sustainable cooperation in finite-horizon  DISG however, needs to take into account the dynamics of the beliefs, since in the repeated games framework this element is absent.
\subsection{Extension to more general reward functions}
\label{discussion_rewards}
The marginal distributions appearing in the expected instantaneous rewards are determined by the belief $\mu^n_t(i^n_t)(\cdot)$. 
Thus, the expected reward function can be expressed in terms of the belief $\mu^n_t(i^n_t)$ instead of $i^n_t$. It then turns out that the analysis and results (except for Corollary \ref{lem_simplex}) of the paper hold for more general bounded functions 
of the form
\begin{align}
\label{general_form}
    r^n(x_t,y^n_t,z^{-n}_t;\mu^n_t(\cdot)),
\end{align}
under some mild conditions ({\em see} $A$, $B$ below). The expected instantaneous reward at time $t$ under action $a^n_t$ becomes
\begin{align*}
    \sum_{x_,y^n_t,z^{-n}_t}\mathbb{P}(x_t,y^n_t,z^{-n}_t\vert i^n_t)r^n(x_t,y^n_t,z^{-n}_t;\mu^n_t(\cdot))-a^n_tc^n.
\end{align*}
Due to the dependence of the reward function $r^n(\cdot)$ on the belief, the above function becomes non-linear in $\mu^n_t(i^n_t)(\cdot)$. Non-linear reward functions incorporating the uncertainty in state estimation are encountered in several fields including controlled sensing \cite{Krishnamurthy_2016}. The DISG model introduced here entails an extra layer of complexity associated with the belief about private information of the other agent.

Lemma \ref{Lem_reward} extends to reward functions of the form \eqref{general_form} as follows.
\begin{align*}
    &\mathbb{E}\{r^n_t(X_t,Y^n_t,Z^{-n}_t)\vert i^n_t\}
    =g^{-n}_t(i^c_t)(a^{-n}_t=1)\nonumber\\
    &\times\sum_{x_,y^n_t,y^{-n}_t}\mathbb{P}(x_t,y^n_t,y^{-n}_t\vert i^n_t)r^n(x_t,y^n_t,y^{-n}_t;\mu^n_t(\cdot)),
\end{align*}
provided that statements $1$ and $2$ of Lemma \ref{Lem_reward} are reinforced with the following conditions
\begin{enumerate}[label=(\Alph*)]
    \item $\sum\limits_{x_t,y^n_t,z^{-n}_t}\mathbb{P}(x_t,y^n_t,z^{-n}_t\vert i^n_t)r^n(x_t,y^n_t,z^{-n}_t;\mu^n_t(\cdot))\geq0$,
    \item $\sum\limits_{x_t,y^n_t}\mathbb{P}(x_t,y^n_t\vert i^n_t)r^n(x_t,y^n_t,Z^{-n}_t=\epsilon;\mu^n_t(\cdot))=0$.
\end{enumerate}
We note that the result of Theorem \ref{Optimal_strategy_finite_horizon} is valid even without conditions $(A)$ and $(B)$, since action $a^n_t=0$ is {\em dominant} for agent $n$ in the static (one-shot) game (note the absence of the action $a^n_t$ in the expression for the expected instantaneous reception gain).

Note that $i^n_t$ as well as $\mu^n_t(\cdot)(X_{0:t},I^{-n,p}_t)$ have a time-increasing domain. %we will use the marginal belief over $X_t$. 
We will show that under the class of Constrained Grim Trigger strategies introduced below, the marginal belief over $X_t$ is a sufficient statistic (in conjunction with another variable defined in the sequel). To avoid confusion we will denote the marginal belief over $X_t$ as 
$\pi_t^n(X_t=x)=
\mathbb{P}^{g^n,g^{-n}}(X_t=x\vert i^n_t)$.
\section{Infinite horizon DISG and Constrained Grim Trigger strategies}
In contrast to the finite horizon and the absolute lack of cooperation, infinite horizon problems may enable the emergence of sustainable cooperation in equilibrium.

{\em Punishment strategies} are a typical example \cite{Tirole, Mailath_2006}. One of the simplest such strategies is the {\em grim trigger} (GT), which in the context of DISG takes the following form for agent $n$: 
\begin{itemize}
\item At time $t=0$ select $a^n_0=1$ (i.e., share $Y^n_0$).
\item For every time $t>0$ select $a^n_t=1$ except if $a^n_{t-1}=0$ or $a^{-n}_{t-1}=0$.
\end{itemize}
Notice that if agent $n$ follows a GT strategy, a single non-cooperative action of agent $-n$ at time $\tau$, results in agent $n$ not cooperating $\forall t>\tau$.\noindent
\subsection{Constrained grim trigger strategies}
Motivated by the above definition, we next introduce the {\em Constrained Grim Trigger} (CGT) strategy. CGT strategies are parametrized by the subsets of the simplex $\Delta(\mathcal{X})$ and are defined over the augmented state space ${\mathcal{S}}\times\Delta(\mathcal{X})$ where $\mathcal{S}=\{0,1\}$ represents the information sharing status. More precisely, let the random variable $S_t:\mathcal{I}^n_t\to \mathcal{S}$ 
that flags the occurrence of deviation from cooperation. Thus, $S_t(i^n_t)=1$ if $\forall n,j< t, a^{n}_j \in i^n_t$, $a^{n}_j = 1$ and $S_t(i^n_t)=0$ otherwise. The dynamics of $S_t$ are deterministic and given by
\begin{align}
\label{s_t-evolution}
\mathbb{P}(s_{t+1}=1|s_t,a^1_t,a^2_t)=\mathds{1}_{\{s_t=a^1_t=a^2_t=1\}}.
\end{align}
\begin{Def}
\label{CIB_GT-def}
Let $\Pi^{n,c}\in\mathcal{P}(\Delta(\mathcal{X}))$ be a subset of the simplex $\Delta(\mathcal{X})$%denote the {\em cooperation region} of agent $n$
, where $\mathcal{P}(\Delta(\mathcal{X}))$ is the powerset of $\Delta(\mathcal{X})$. The {\em Constrained Grim Trigger} (CGT) strategy is defined as follows. Let $\mathcal{F}_\mathcal{X}$ denote the space of mappings $\sigma: \{0,1\} \times \Delta(\mathcal{X}) \to \Delta(\mathcal{A})$. Define the CGT mapping  $\mathbb{\sigma}^{n,\cdot}:\mathcal{P}(\Delta(\mathcal{X}))\to \mathcal{F}_\mathcal{X}$ for agent $n$, by
\begin{align*}
\sigma^{n,\Pi^{n,c}}(s_t,\pi^n_t)(a_t^n=1)=
\begin{cases}
1, &\text{if} \, \, \ s_t =1 \, \,  \text{and} \, \,  \pi^n_t \in \Pi^{n,c},\\
0, &\text{otherwise},
\end{cases}
\end{align*}%}}
where $\pi^n_t$ is the belief over system states with elements $\pi^n_t(X_t=x)=\mathbb{P}^{\sigma^{n,\Pi^{n,c}},\sigma^{-n,\Pi^{-n,c}}}(X_t=x|i^n_t), x\in\mathcal{X}$. The elements of the image of $\sigma^n$ are called {\em CGT strategies} for agent $n$. $\Pi^{n,c}$ symbolizes the {\em cooperation region} of strategy $\sigma^{n,\Pi^{n,c}}$. 
\end{Def}
A CGT strategy $\sigma^{n,\Pi^{n,c}}$ declares that agent $n$ shares information as long as her belief $\pi^n_t$ lies in the region $\Pi^{n,c}$ (hence ``Constrained'') and both agents shared information at every time instant up to the current epoch (i.e., $S_t=1$). It can be seen that each CGT strategy, is uniquely defined by an element of $\Pi^{n,c}\in \mathcal{P}(\Delta(\mathcal{X}))$; the CGT mapping for each agent is injective.

A CGT strategy for agent n corresponding to an arbitrary belief set $\Pi^{n,c}$ is a stationary deterministic mapping and can be written as 
{\small{\begin{align}
\label{sigma_ind1}
&\sigma^{n,\Pi^{n,c}}(s_t,\pi^n_t)(a^n_t=1)=\mathds{1}_{\{s_t=1,\pi^n_t\in\Pi^{n,c}\}}.%
\end{align}}}
In the sequel, we write $\sigma^{-n}$ instead of $\sigma^{-n,\Pi^{-n,c}}$ whenever it is clear from the context.

Under CGT strategies, two distinct phases can exist during agents' interactions. The first phase consists of full data exchange. During this phase there is no private information. The second phase initiates after a deviation from cooperation occurs and during that phase agents' observations constitute private information. 

In the sequel, we examine agents' optimal behavior 
under CGT strategies in the infinite horizon DISG. 
\begin{Lem}
\label{sigma_common}
If agent $-n$ follows a CGT strategy $\sigma^{-n,\Pi^{-n,c}}$, the following statements hold:
\begin{enumerate}
\item\noindent\vspace{-0.57cm}
\begin{align}\noindent
\label{sigma_inference}\noindent
&\mathbb{P}^{\sigma^{-n}}(a^{-n}_t|i^{-n}_t)=\sigma^{-n}(s_t,\pi^{-n}_t)(a^{-n}_t)\nonumber\\
&=\mathbb{P}^{\sigma^{-n}}(a^{-n}_t|i^n_t)=\sigma^{-n}(s_t,\pi^n_t)(a^{-n}_t), \quad\forall t.
\end{align}
\item Agent $n$'s belief $\pi^n_t$ is updated recursively as $\pi^n_{t+1}=f(\pi^n_t,y^n_t,z^{-n}_t,a^{-n}_t)$.
\item Agent $n$'s reward function for given $i^n_t$ and action $a^n_t$, is given by
\begin{align}
\label{GT_Reward}
&\tilde{R}^n(s_t,\pi^n_t,a^n_t)
=\tilde{r}^n(s_t,\pi^n_t)-a^n_tc^n,
\end{align}
\end{enumerate}
where
\begin{align}
\label{GT_Reward_aux}
&\tilde{r}^n(s_t,\pi^n_t)=\mathds{1}_{\{s_t=1,\pi^n_t\in\Pi^{-n,c}\}}I(X_t;Y^{-n}_t|Y^n_t,i^n_t).
\end{align}
\end{Lem}
The next theorem states that if agent $-n$ follows a CGT strategy, then agent $n$ faces a POMDP with information state $(s_t,\pi^n_t)$. Hence, agent $n$ can choose her best-response from the class of strategies that depend on $(s_t,\pi^n_t)$ without loss of optimality, because in infinite horizon POMDPs stationary
strategies that depend on the information state are optimal.  We will further show that the CGT strategies
are closed under the best response mapping, meaning that if agent $-n$ follows a CGT strategy, then agent $n$ can optimally respond using a CGT strategy.
\begin{Thm}
\label{ST-POMDP}
Given that agent $-n$ follows a CGT strategy, %$\sigma^{-n,\Pi^{-n,c}}$
agent $n$'s best-response problem is a POMDP. Moreover, $(s_t,\pi^n_t)$ is an {\em information state}.
\end{Thm}
Since, agent $n$'s best-response problem corresponds to a POMDP, the {\em Bellman Equation} (BE) holds:
\begin{align}
\label{BE_infinite_horizon}
&V^n(s,\pi^n)=\max_{a^n\in\{0,1\}}\{\tilde{r}^n(s,\pi^n)-a^nc^n\nonumber\\
&+\delta \mathbb{E}^{\sigma^{-n}}\{V^n(s',f(\pi^n,y^n,z^{-n},a^{-n}))|\pi^n,s\}\},
\end{align}
where $s'$ stands for the future value of $s$ and 
$\tilde{r}(s,\pi^n)$ is given by \eqref{GT_Reward_aux}. The expectation is w.r.t. all random variables and is computed as
\begin{align}
\label{continuation_value}
&\mathbb{E}^{\sigma^{-n}}\{V^n(s',f(\pi^n,y^n,z^{-n},a^{-n}_t
)|\pi^n,s\}\nonumber\\
&=\sum_{y^n,z^{-n},s',a^{-n},x}\mathbb{P}(s'|s,a^n,a^{-n})
\nonumber\\
&\times\mathbb{P}(y^n|x)\mathbb{P}(z^{-n}|x,a^{-n})\sigma^{-n}(s,\pi^n)(a^{-n})
\nonumber\\
&\times \pi^n(x)V^n(s',f(\pi^n,y^n,z^{-n},a^{-n}
)),
\end{align}
where 
$\mathbb{P}(z^{-n}|x,a^{-n})$ is given by \eqref{y_z}.

Eq. \eqref{BE_infinite_horizon} expresses the total expected sum of discounted rewards for agent $n$ starting from state $s,\pi^n$, given that agent $-n$ follows a CGT strategy $\sigma^{-n}(s,\pi^n)$ and agent $n$ acts optimally. For $s=0$, \eqref{BE_infinite_horizon} yields
\begin{align}
\label{BE_s0a}
&V^n(s=0,\pi^n)=\max_{a^n_{0:\infty}}\sum^{\infty}_{t=0}\delta^t\nonumber\\
&\times\mathbb{E}^{\sigma^{-n}}\{\tilde{r}^n(s=0,\pi^n)-a^nc^n|\pi^n,s=0\}.
\end{align}
For $s=0$, it is $\sigma^{-n}(s=0,\pi^{-n})(a^{-n}=0)=1$ and as a result $\tilde{r}^n(s=0,\pi^n_t)=0$ ({\em see} \eqref{GT_Reward}, \eqref{GT_Reward_aux}) and $s_{t+1}=0$ for every $t$ ({\em see} \eqref{s_t-evolution}). Thus, \eqref{BE_s0a} yields
\begin{align}
\label{BE_s0b}
&V^n(s=0,\pi^n)=\max_{a^n_{0:\infty}}\sum^{\infty}_{t=0}-\delta^tc^na^n_t,
\end{align}
which clearly takes the maximum value when $a^n_t=0$ for all $t,\pi^n$. So, for $s=0$, the only sequentially rational strategy for agent $n$ is to select $a^n_t=0$ for all $t$ and then, \eqref{BE_s0b} gives
\begin{align}
\label{BE_s0c}
&V^n(s=0,\pi^n)=0, \quad\forall \pi^n.
\end{align}
The expected future rewards for agent $n$ for a given state action pair are given by
\begin{align}
\label{Q-function}
&Q^n(s,\pi^n,a^n)=\tilde{r}(s,\pi^n)-a^nc^n+\delta\\
&\times\mathbb{E}^{\sigma^{-n}
}\{V^n(s',f(\pi^n,y^n,z^{-n}
,a^{-n}))|\pi^n,s,a^n\}.\nonumber
\end{align}
Utilizing 
\eqref{GT_Reward_aux}, \eqref{BE_s0c}, \eqref{Q-function}, we obtain for every $\pi^n$
\begin{align}
\label{Q_1}
&Q^n(s=0,\pi^n,a^n)=-a^nc^n,\\
\label{Q_2}
&Q^n(s=1,\pi^n,a^n=0)=\tilde{r}(1,\pi^n),\\
\label{Q_3}
&Q^n(s=1,\pi^n,a^n=1)=\tilde{r}(1,\pi^n)-c^n+\\
&\delta\mathbb{E}^{\sigma^{-n}}\{V^n(s',f(\pi^n,y^n,z^{-n},a^{-n}))|\pi^n,s=1,a^n=1\}.\nonumber
\end{align}
\begin{Thm}
\label{Closeness}
The CGT strategies are closed under the best-response mapping.
\end{Thm}
Next we demonstrate an important feature of CGT strategies: they give rise to PBEs that can be grouped up into equivalence classes, which are characterized by the strategy profile and $\pi_t^n$ for each $n$. This allows us to ignore the belief on other agent's private information. Thus, despite the fact that private information is present in the DISG, the PBE solution concept becomes redundant when one considers equilibria consisting of CGT strategies; it will be enough to consider Subgame Perfect Equilibria (SPEs). 
To show the following result, it will be convenient to define the following marginalization operator $\pi^{n,X}(\mu)=\{\pi^{n,\mu}_t\}_{t\in\mathbb{N}}$, where
\begin{align}
	\pi^{n,\mu}_t(i^n_t)(X_t=x_t) = \sum\limits_{i^{-n,p}_t,x_{0:t-1}}\ \mu^n_t(i^n_t)(x_{0:t},i^{-n,p}_t).
\end{align}
Note that $\pi_t^n(X_t=x)=\pi_t^{n,\mu}(i^n_t)(X_t=x)=\mathbb{P}^{g^n,g^{-n}}(X_t=x\vert i^n_t)$.
\begin{Thm}
\label{pbequiv}
Suppose $(\sigma^*,\mu)$ is a PBE, such that $\sigma^*$ is a CGT profile. Then, $(\sigma^*,\mu')$ where $\mu'$ is a consistent belief profile w.r.t. $\sigma^*$ such that $\pi^{n,X}(\mu')=\pi^{n,X}(\mu)$ is also a PBE. 
\end{Thm}
\begin{Rem}
\label{rem_spe}
Theorem \ref{pbequiv} states that in order to check whether a pair of CGT strategies are sequentially rational, beliefs on past states and other agent's private information are irrelevant. 
\end{Rem}
\noindent\subsection{Equilibrium regions}
Let $V^{n,C,C'}$ denote the value function 
of agent $n$ under the strategy profile $(\sigma^{n,C},\sigma^{-n,C'})$ and let $V^{n,*,C}$ denote the optimal value function of agent $n$ when agent $-n$ follows $\sigma^{-n,C}$. Similarly,  $Q^{n,*,C}(s,\pi^n,a^n)=\tilde{r}(s,\pi^n)-a^nc^n+\delta\mathbb{E}\{V^{n,*,C}(s',\pi'^n)\vert s,\pi^n,a^n\}$. 
We also define the operator $O^n(\cdot):\mathcal{P}(\Delta(\mathcal{X}))\to\mathcal{P}(\Delta(\mathcal{X}))$, as 
\begin{align}
&O^n(C) = \big\{\pi \in \Delta(\mathcal{X}) \  \vert\ Q^{n,*,C}(s=1,\pi,a^n=1)\nonumber\\
&\geq Q^{n,*,C}(s=1,\pi,a^n=0)\big\}.
\end{align}
In words, $O^n(C)$ may be thought as an oracle for the POMDP that agent $n$ has to solve to get the optimal CGT strategy when agent $-n$ follows a CGT strategy with cooperation region $C$. Note that such an optimal strategy for agent $n$ exists from Theorems \ref{ST-POMDP} and \ref{Closeness}. Also note that since $O^n(C)$ corresponds to the solution of the aforementioned POMDP, it is determined by the primitives of the problem,  $c^n,\delta$, the system dynamics and the agents' observation models.
\begin{Def}
A pair of regions $(\Pi^{1,c},\Pi^{2,c})\in\mathcal{P}(\Delta(\mathcal{X}))\times\mathcal{P}(\Delta(\mathcal{X}))$ is in {\em cooperation equilibrium}, if
\begin{align}
\label{equilibrium_region}
\Pi^{n,c}=O^n(\Pi^{-n,c}), \quad n\in\{1,2\}.
\end{align}
\end{Def}
The following Proposition characterizes regions that are in cooperation equilibrium.
\begin{Prp}
\label{Prp_regions}
The following statements are true:
\begin{enumerate}
    \item $\forall\ C\in \mathcal{P}(\Delta(\mathcal{X}))$, $O^n(C)\subseteq C$, $n=\{1,2\}$.
    \item If a pair of regions $(\Pi^{1,c},\Pi^{2,c})$ is in {\em cooperation equilibrium}, then the two regions coincide, that is $\Pi^{1,c}=\Pi^{2,c}$.
\end{enumerate}
\end{Prp}
In light of part $2)$ of Proposition \ref{Prp_regions}, we say that a region $\Pi^c \in \mathcal{P}(\Delta(\mathcal{X}))$ is an {\em equilibrium region} if the pair $(\Pi^c,\Pi^c)$ is in cooperation equilibrium. Let $\mathcal{E}\subseteq \mathcal{P}(\Delta(\mathcal{X}))$ be the set of all equilibrium regions.
\begin{Rem}
Regarding part $2)$ of Proposition \ref{Prp_regions}, we note that the intuition behind this result is the following. It is never favorable for an agent to cooperate in regions of the belief simplex that the other agent will not cooperate for $c^n>0$. For example, in the extreme case when $c^1\rightarrow\infty$ and $c^2\rightarrow 0$, agent $1$ will not cooperate (since $\tilde{r}^n(s,\pi^n)$ is bounded), and thus the other agent will not cooperate either, since she has no gain and pays a small positive cost if she does.
\end{Rem}
\begin{Prp}
\label{CGT_PBE}
The strategy profile $\sigma^*=(\sigma^{1,\Pi^c},\sigma^{2,\Pi^c})$ where $\Pi^c\in\mathcal{E}$, is a SPE.
\end{Prp}
\begin{Thm}
\label{Value_function_increasing_c}
Let $C \subseteq C' \subseteq \Delta(\mathcal{X})$. Then, the following hold:
\begin{enumerate}
\item The value function of agent $n$ is non-decreasing in the other agent's cooperation region. That is,
\begin{align}
V^{n,*,C}(s,\pi) \leq  V^{n,*,C'}(s,\pi), \quad\forall s,\pi.
\end{align}
\item Let $\pi \in C$. Then,
\begin{align}
\label{BE_increasing}
E\{V^{n,*,C}(s'=1,\pi')\vert \pi, s=1, a^n=1\}\leq\nonumber\\ E\{V^{n,*,C'}(s'=1,\pi')\vert \pi, s=1, a^n=1\}.
\end{align}
\end{enumerate}
\end{Thm}
\begin{Lem}
\label{Lem_7}
Let $C$ be an equilibrium region. Then $\forall \pi \in C$ the following inequality holds:
\begin{align}
\delta E\{V^{n,*,C}(s'=1,\pi')\vert \pi, s=1,a^n=1\} - c^n \geq 0.
\end{align}
\end{Lem}
It is trivially seen from the definition of a CGT strategy that the empty set $\emptyset$ is always an equilibrium region and thus $\emptyset\subset\mathcal{E}$. Furthermore, $\mathcal{E}$ is a partially-ordered set under set inclusion and it is easily seen that any chain $C_1 \subseteq C_2 \subseteq \dots$ where $\forall i,\ C_i\in\mathcal{E}$ has an upper bound (namely $\bigcup\limits_{i}^{\infty} C_i$). Thus, by Zorn's Lemma there exists at least one maximal element. We now argue that in fact there exists a {\em unique maximal element} which we will call it the {\em maximal equilibrium region}.
\begin{Thm}
\label{Th3}
There exists a unique maximal equilibrium region $\Pi^* \in \mathcal{E}$.  Moreover, the strategy profile $(\sigma^{n,\Pi^*},\sigma^{-n,\Pi^*})$ is optimal in the sense that $V^{n,\Pi^*,\Pi^*}(s,\pi) \geq V^{n,C,C}(s,\pi^n),\ \forall C \in \mathcal{E},n$.
\end{Thm}
Next we describe a theoretical algorithmic scheme for calculating $\Pi^*$. For simplicity define the operator $$F^n(C)=O^n(O^{-n}(C)).$$
Clearly $\mathcal{E}$ is also the set of all fixed points of $F$. Also it follows from part $1)$ of Proposition \ref{Prp_regions} that $F^n(C)\subseteq C$, $\forall C\subseteq \mathcal{P}(\Delta(\mathcal{X}))$.\\
\noindent\rule{\linewidth}{0.5mm} \\[-0.5mm]
\textbf{Iterative Refinement Algorithm (ItRA)}\\[-2mm]
\rule{\linewidth}{0.5mm}
Input: $k$ (number of iterations), $\Pi^{n,c}=\Delta(\mathcal{X})$
\begin{itemize}
\item for $k$ iterations do:
\begin{itemize}
\item $\Pi^{n,c}\leftarrow F^n(\Pi^{n,c})$
\item if $\Pi^{n,c} = F^n(\Pi^{n,c})$, then halt and return $\Pi^{n,c}$
\end{itemize}
\item Return $\Pi^{n,c}$
\end{itemize}
\rule{\linewidth}{0.5mm}\\[-0.5mm]
The following result states that the operator $O^n(C)$ always contains the maximal equilibrium region $\Pi^*$ and as a result, $\Pi^* \subseteq ItRA(k)$ for every $k>0$.
\begin{Prp}
\label{Lem_9}
If $\Pi^*\subseteq C$, then the following are true:
\begin{enumerate}
    \item $\Pi^*\subseteq O^n(C)$.%, ItRA().
    \item $\forall k>0$, $\Pi^* \subseteq ItRA(k)$ and $ItRA(k+1)\subseteq ItRA(k)$.
\end{enumerate}
\end{Prp}%
Note that the above result implies that if ItRA halts early, then the computed region is $\Pi^*$. If not, the algorithm computes an upper bound that becomes finer as $k$ increases. We wish to point out that the algorithm (as well as the rest of our results except for Corollary \ref{lem_simplex}) is applicable to setups with more general reward functions of the form discussed in Section \ref{discussion_rewards}.

To prove that cooperation can indeed be sustained in the infinite horizon, it remains to show that there are appropriate choices of parameters $c^n, \delta$ for which the maximal region $\Pi^*$ is non-empty. To this end, we give the following definition.
\begin{Def}
For given state transition and observation models of the agents, we say a non-empty set $C\subseteq \Delta(\mathcal{X})$ is {\em absorbing} if it holds that, if $\pi \in C$ , we have $\pi'=f(\pi,y^n,y^{-n},a^{-n}=1) \in C$, for all observations $y^n\in\mathcal{Y}^n,y^{-n}\in\mathcal{Y}^{-n}$. 
Moreover, a {\em positive absorbing set} $C$ is an absorbing set for which
\begin{align}
\label{rineq}
    &r^{C}_{inf}=\min_{n\in\{1,2\}}\underset{\pi\in C}\inf\{\tilde{r}^n(s=1,\pi)\} > 0.
\end{align}
\end{Def}
Note that if $C$ is absorbing and $\pi_t^n \in C$, then $\pi^n_{t+k} \in C$ for every $k \in \mathbb{N}$ as long as $s_{t+k}=1$. In other words, the set $C$ {\em traps} the belief, in the sense that while no deviation from cooperation has taken place up to time $t$, and, the common belief of the agents lies in $C$ at $t$, then the common belief of the agents will continue to lie in $C$ as long as agents continue to share their observations.

\begin{Thm}
\label{bound_region}
The following are true:
\begin{enumerate}
\item Given a discount factor $\delta\in(0,1)$, a state transition kernel and the observation models for the agents, a positive absorbing set is an equilibrium region given that 
\begin{align}
\label{condition_Theorem_7_1}
    & c^n = \tilde{\epsilon}
    r^{C}_{inf},\quad n=1,2,
\end{align}
where $\delta\geq\tilde{\epsilon} >0$.
\item Let us define
    \begin{align}
    \label{rineqr}
        &\lambda_{min}(x')=\underset{x\in X}\min\{\mathbb{P}(x'\vert x)\}\\
        &\lambda_{max}(x')=\underset{x\in X}\max\{\mathbb{P}(x'\vert x)\}\\
        &\Lambda = \bigotimes\limits_{x'\in X} [\lambda_{min}(x'), \lambda_{max}(x')].
    \end{align}
    The region $C= \Lambda \cap \Delta(\mathcal{X})$ is absorbing. Further, if it is positive absorbing, then given a discount factor $\delta\in(0,1)$, $\exists$ $c^n$, $n=1,2$  such that $\Pi^* = \Delta(\mathcal{X})$.
\end{enumerate}
\end{Thm}
The above result gives us a means to prove lower bounds for $\Pi^*$ and hence non-emptiness. In particular, identifying a positive absorbing set for the problem at hand gives us an equilibrium region and proves that cooperation can be sustained in the infinite horizon. Moreover, we show that at least one absorbing set always exists and by ensuring this is also positive absorbing, 
we get that $\Pi^*$ is non-empty. In the following result we utilize the mutual information utility function to provide conditions 
under which this set is positive absorbing.
\begin{Cor}
\label{lem_simplex}
Let $C= \Lambda \cap \Delta(\mathcal{X})$ as above and let $$\pi^n_{min} = \underset{\pi\in C}{\arg\min}\{\tilde{r}^n(s=1,\pi)\},\quad n=1,2. $$
If $X\in\mathcal{X}$ and $Y^{-n}\in\mathcal{Y}^{-n}$ are conditionally dependent given $Y^n\in\mathcal{Y}^n, \pi^n_{min}$ for $n=1,2$ 
then $C$ is positive absorbing.
\end{Cor}
The result is true due to the following. Note that $C$ is compact as an intersection of compact sets and hence the minimum over $C$ is well defined since $\tilde{r}$ is continuous. If $X$ and $Y^{-n}$ are conditionally dependent given $Y^n,\pi^n_{min}$ for $n=1,2$, then $r^C_{inf}>0$, since
\begin{align}
\label{m_inf}
    \tilde{r}^n(s=1,\pi^n_{min})=0\Leftrightarrow
    I(X;Y^{-n}\vert Y^n%,\pi_{min}
    )=0,
\end{align}\noindent
is true if and only if $X$ and $Y^{-n}$ are conditionally independent given $Y^n,\pi^n_{min}$ \cite{Gallager}. The distribution of state $X$ in \eqref{m_inf} is given by $\pi^n_{min}$.

Note that if $X$ and $Y^{-n}$ are conditionally independent given $Y^n, \pi^n$, this implies that no information is conveyed from $Y^{-n}$ about state $X$. For instance, this can happen if the observation model of agent $n$ is fully informative (i.e., deterministically reveals the state $X$) or if $Y^{-n}$ is {\em uninformative} (i.e., $\mathbb{P}(y^{-n}\vert x)=\mathbb{P}(y^{-n}\vert x')$ for all $x\neq x'$ and for all $y^{-n}\in\mathcal{Y}^{-n}$). 
Note also that if $\pi^n$ is a vertex of $\Delta(\mathcal{X})$ then $X$ and $Y^{-n}$ are conditionally independent given $\pi^n$. This implies that the transition kernel must be positive (all elements strictly greater than $0$) for Corollary \ref{lem_simplex} to hold.
\begin{Rem}
Part $1)$ of Theorem \ref{bound_region} is intuitively linked to the economic literature on repeated games \cite{Mailath_2006}. Given a positive absorbing set $C$, define a repeated game $\mathcal{L}$ with payoff matrix given by Table \ref{payoff_matrix}.  
If $\delta,c^n$ for $n=1,2$, 
\begin{table}[!h]\small
\caption{Payoff matrix of the game $\mathcal{L}$. In each entry of the table, the first payoff corresponds to the row agent $n=1$, while the second payoff corresponds to the agent $n=2$. $r^C_{inf}$ is given by \eqref{rineq}.}\label{payoff_matrix}
\centering
\begin{tabular}{| c | c | c |}\hline
\mybackslash{$\,A^1_t$\ }{$A^{2}_t$} & $0$ & $1$ \\\hline\hline
$0$ & $(0,0)$ & $(r^C_{inf},-c^2)$ \\\hline
$1$ & $(-c^1,r^C_{inf})$ & $(r^C_{inf}-c^1,r^C_{inf}-c^2)$ \\\hline
\end{tabular}
\end{table}
 are such that under the classical GT strategy  cooperation is sustained in $\mathcal{L}$, then for such $\delta,c^n$,  cooperation is also sustained in DISG, if  $\tilde{\pi}_0\in C$, under the CGT strategy where the agents cooperate in $C$. This is because for both agents the payoffs associated with cooperation in DISG are greater or equal than the ones in $\mathcal{L}$. By standard results for repeated games applied in $\mathcal{L}$ \cite{Mailath_2006}, cooperation is sustained under classical GT if $
 \delta\geq\frac{c^n}{r^C_{inf}}$, which is equivalent to \eqref{condition_Theorem_7_1}.
\end{Rem}
\subsection{Experiments}
The purpose of the experiments discussed next is to empirically illustrate the existence of equilibrium regions in the infinite horizon DISG and demonstrate that cooperation is sustainable. 
We assume that the Markov chain entails a binary state and that each agent $n\in\{1,2\}$ has access to a binary symmetric channel (BSC) with observation probabilities parametrized by $p_1=\mathbb{P}(Y^1=0\vert X=0)=\mathbb{P}(Y^1=1\vert X=1)$ and $p_2=\mathbb{P}(Y^2=0\vert X=0)=\mathbb{P}(Y^2=1\vert X=1)$. In this setup, we investigate the effect of these parameters and the communication cost on the cooperation region.

We use an {\em online planning algorithm} to solve the POMDP that corresponds to agent's best-response problem. In particular, in our implementation we used a slight modification of the POMCP algorithm \cite{Silver_2010} to approximate the operator $O^n(\cdot)$ in the ItRA algorithm. The particle filter used in the POMCP was replaced with the exact belief update, to be able to calculate the rewards of the agents while running simulations. Moreover, we discretize the belief simplex using a fine grid on which the optimal actions are computed. The results are subject to  approximation errors due to the approximate nature of the POMCP, which employs simulated averages instead of expectations, and the discretization grid.

For our experiments, we approximate the optimal equilibrium region for different values of $p_1,p_2,c$ (we set $c^1=c^2=c$). 
The state transition probabilities are given by $\mathbb{P}(X_{t+1}=0\vert X_t=0)=0.8$, $\mathbb{P}(X_{t+1}=1\vert X_t=0)=0.2$, $\mathbb{P}(X_{t+1}=0\vert X_t=1)=0.15$, $\mathbb{P}(X_{t+1}=1\vert X_t=1)=0.85$. Regarding the agents' emission probabilities, the BSCs are parametrized by  $p_1=p_2=0.6$. 
The rest of the parameters are set to $\delta=0.9, c=0.027$.

In Fig. \ref{g2} we depict the cooperation regions for three different parameter setups, as computed by the aforementioned scheme. The $x$-axis represents the belief $\pi^n(X=0)$. We observe ({\em see} top line in Fig. \ref{g2}) that cooperation is sustainable in a subset of the belief simplex $\Delta(\mathcal{X})$.

Next, we show the impact of the transmission cost on the cooperation region by changing $c=0.027$ to $c=0.024$. As we observe ({\em see} bottom line in Fig. \ref{g2}), the cooperation region gets larger for smaller transmission cost. This confirms intuition, because cooperation becomes less expensive and thus, agents opt to share information in a larger subset of $\Delta(\mathcal{X})$.

Finally, we show the impact of the observation probabilities in agents' optimal policy. We assume that agent $1$ has more `qualitative' observations (by means of being more discriminating between the two states and thus, providing smaller uncertainty over the system state) and change $p_1$ from $0.6$ to $0.65$, while keeping the cost at $c=0.024$. We observe ({\em see} the middle line in Fig. \ref{g2}) that the cooperation region becomes smaller (compared with the bottom line of Fig. \ref{g2}), as agent $1$, now has less incentives to cooperate and acquire information from agent $2$.

\begin{figure}[t]
\centering
\includegraphics[width=0.45\textwidth,height=0.25\textwidth]{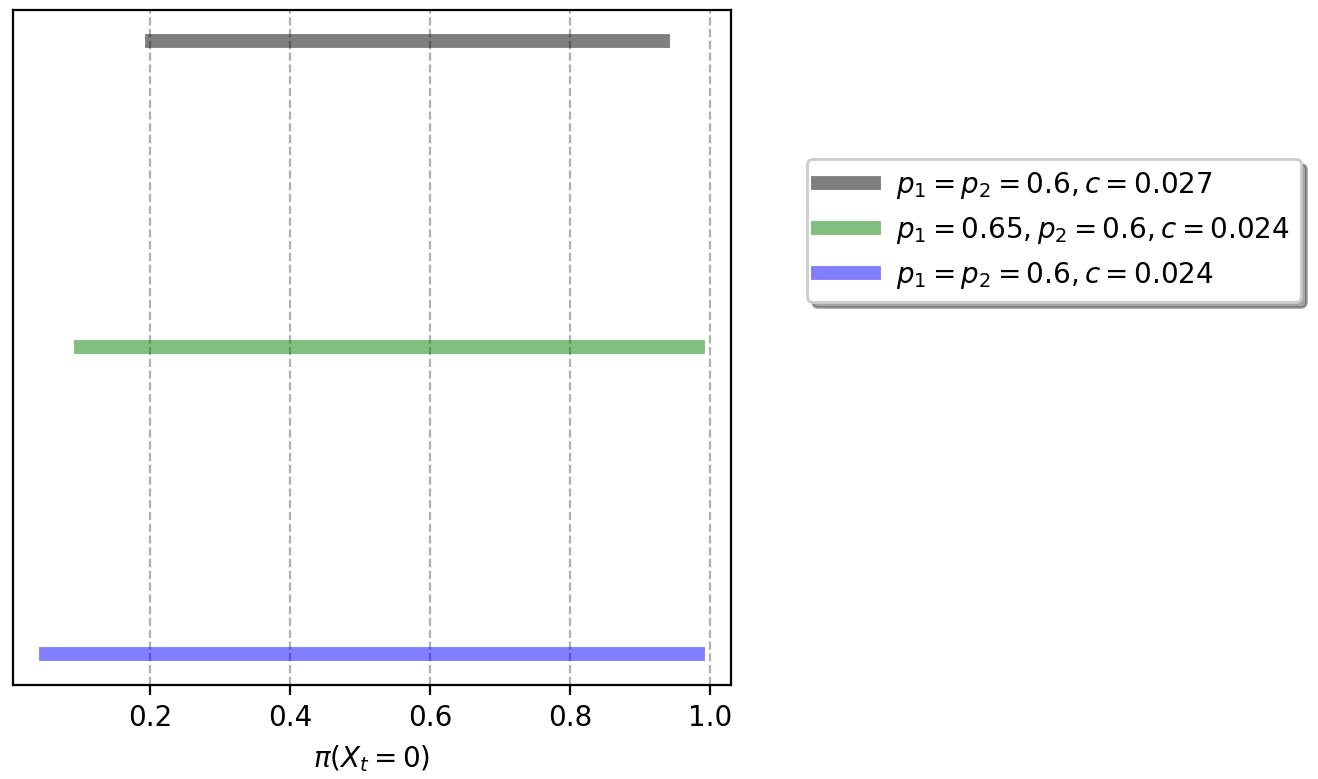}
\caption{Visual illustration of the computed cooperation region for different parameter values.}
\label{g2}
\end{figure}

\section{Conclusion}
In this work, the information sharing process between two rational selfish agents interested in an estimation task was studied. We employed the conditional mutual information to quantify the value of information exchanged between the agents. We showed that in the finite horizon DISG, cooperation can not emerge at equilibrium. This led us to consider CGT strategies to check whether cooperation can be sustained in the infinite horizon setting. We showed that these strategies are closed under the best-response mapping and that cooperation can emerge at equilibrium. Finally, we characterized the equilibrium regions, proved uniqueness of a maximal equilibrium region, devised an iterative algorithm whose output provably contains it and provided results that ensure its non-emptiness.

The proposed model and results have potential to provide useful insight in consensus or diffusion based distributed networks whose nodes perform estimation, detection, control or routing tasks and combine individual sensing data with signals received from neighbors. Another problem that is interesting to explore utilizing the ideas developed in this work is Bayesian learning and the study of information cascades \cite{Bikhchandani_1992}. Recently, this problem has been studied with agents acting sequentially, instead of the static case where agents act only once \cite{Vasal_2016}. The ideas developed in this paper could be utilized to investigate the potential to avoid inefficient information cascades.

The DISG model may contribute to the above research areas by endogenizing the information sharing decision. The DISG model could also be applicable to the study of networks with adversarial nodes where the received information might be meaningful, irrelevant, or malicious.

Taking full advantage of the DISG model requires additional work regarding three assumptions made in the paper: (i) information sharing takes place between two agents, (ii) agents have the option to share only the acquired observations instead of sharing arbitrary information (they do not have the option to ``lie"), (iii) CGT strategies are sensitive to errors. Models of multiple agent interactions and more general constrained strategies under noisy transmissions are a subject of ongoing research.

\section*{Acknowledgment}
The authors would like to thank Professor Nicholas Kolokotronis, Associate Professor in Department of Informatics and Telecommunications, University of Peloponnese, Greece, for the fruitful discussions. They would also like to thank the reviewers for their helpful comments.
\appendix  % for no appendix heading
\section*{Proof of Lemma \ref{Lem_reward}}
The expected reception gain function \eqref{expected_instantaneous_reward} yields
\begin{small}
\begin{align}
\label{reward_aux}
&\sum_{x_t,y^n_t,Z^{-n}_t=\epsilon}\mathbb{P}(Z^{-n}_t=\epsilon\vert x_t,i^n_t)\mathbb{P}(y^n_t\vert x_t)\mathbb{P}(x_t\vert i^n_t)\nonumber\\
&\times r^n_t(x_t,y^n_t,Z^{-n}_t=\epsilon,i^n_t)+\sum_{x_t,y^n_t,Z^{-n}_t=y^{-n}_t}\mathbb{P}(Z^{-n}_t=y^{-n}_t\vert x_t,i^n_t)\nonumber\\
&\times\mathbb{P}(y^n_t\vert x_t)\mathbb{P}(x_t\vert i^n_t)r^n_t(x_t,y^n_t,Z^{-n}_t=y^{-n}_t,i^n_t).
\end{align}
\end{small}\noindent
Suppose $i^n_t=i^{-n}_t=i^c_t$. Then, \eqref{ins_reward2}, \eqref{infer_action} imply $P(Z^{-n}_t=\epsilon|x_t, i^n_t)=g^{-n}(i^c_t)(a^{-n}_t=0)$ and the latter expression does not depend on $x_t$. Then, it is easy to verify from \eqref{instantaneous_reward}, \eqref{ins_reward} that the first term of the summation in \eqref{reward_aux} is equal to $0$. Eq. \eqref{ins_reward3}, \eqref{infer_action} imply $P(Z^{-n}_t=y^{-n}_t|x_t,
i^n_t)=g^{-n}(i^c_t)(a^{-n}_t=1)P(y^{-n}_t|x_t)$. Replacing it into \eqref{reward_aux} yields
\begin{align}
&\mathbb{E}\{r^n_t({X}_t,Y^n_t,Z^{-n}_t)|I^n_t=i^n_t\}\nonumber\\
&=g^{-n}_t(i^c_t)(a^{-n}_t=1)I(X_t;Y^{-n}_t|Y^n_t,i^n_t).
\end{align}
To prove part $2$ of the Lemma, note that if $g^{-n}_t(i^{-n,p}_t,i^c_t)(a^{-n}_t)=g^{-n}_t(\bar{i}^{-n,p}_t,i^c_t)(a^{-n}_t)$ for every $i^{-n,p}_t\neq\bar{i}^{-n,p}_t$, then \eqref{infer_action} yields
\begin{align}
&\mathbb{P}(a^{-n}_t=a|x_t,i^n_t)=g^{-n}_t(i^{-n}_t)(a)\nonumber\\
&\times\sum_{i^{-n,p}_t}\frac{\mathbb{P}(x_t|i^{-n,p}_t,i^n_t)\mathbb{P}(i^{-n,p}_t|i^n_t)}{\mathbb{P}(x_t|i^n_t)}=g^{-n}_t(i^{-n}_t)(a).
\end{align}
Then, by working as in part 1, \eqref{MI_rem} is obtained.

\section*{Proof of Theorem \ref{Optimal_strategy_finite_horizon}}
Let $T$ denote the horizon length. Denote,
\begin{equation*}
J^{n,g^{n,*},g^{-n,*}}_{\mu}(i^n_t,t)=\mathbb{E}^{g^{n,*},g^{-n,*}}_{\mu^n}\{\sum^T_{j=t}R^n_t(X_j,Y^n_j,Z^{-n}_j,A^n_j)\vert i^n_t\}.
\end{equation*}
 Suppose $(g^*,\mu)$ is a PBE. We show by strong induction on $k\in \mathcal{T}$ that $g^{n,*}_{T-k}(i^n_{T-k})(a^n_{T-k}=1) \equiv 0, \ \forall k,n,i^n_{T-k}$ and thus that $g^*\equiv g^{NC}$. This proves the result, since $T$ was chosen arbitrarily. 
 For $k=0$, the sequential rationality condition \eqref{sequential_rationality}, implies that $\forall n, i^n_T$
\begin{subequations}
\begin{align}
&J^{n,g^{n,*},g^{-n,*}}_\mu(i^n_T,T) = \underset{g^n_T}{\sup}\  J^{n,g^{n},g^{-n,*}}_{\mu}(i^n_T, T)\nonumber\\
&=\underset{g^n_T}{\sup}\ \mathbb{E}_{\mu^n}\{\mathbb{E}^{g^{-n,*}}_{i^{-n}_T}\{ r^n_T\vert i^n_T\}\}- \mathbb{E}^{g_T^n(i^n_T)}\{a^n_T\vert i^n_T\} c^n \nonumber\\
\label{subequation}
&=\underset{g^n_T}{\sup}\ \mathbb{E}_{\mu^n}\{\mathbb{E}^{g^{-n,*}}_{i^{-n}_T}\{ r^n_T\vert i^n_T\}\} - g_T^n(i^n_T)(a^n_T =1) c^n \\
&=\mathbb{E}_{\mu^n}\{\mathbb{E}^{g^{-n,*}}_{i^{-n}_T}\{ r^n_T\vert i^n_T\}\}.
\end{align}
\end{subequations}
The supremum is attained when $\forall n, i^n_T$ we have that $g_T^n(i^n_T)(a^n_T =1)=0$, because the first term in \eqref{subequation} does not depend on $g^n_t$.

Now, suppose that for $j\leq k$ it holds that $g^{n,*}_{T-j}(i^n_{T-j})(a^n_{T-j}=1)\equiv 0$, for all $j,n,i^n_{T-j}$. By the induction hypothesis the strategies from time $T-k$ onwards are independent of the agents' private information. Hence, by  part 2 of Lemma \ref{Lem_reward}, the expected instantaneous rewards $\forall t \geq T-k$ are $0$ and as a result, the expected sum of payoffs from time $T-k$ onwards is $0$, as well. Hence, by the sequential rationality condition, $\forall n,i^n_{T-k-1}$
\begin{align}
&J^{n,g^{n,*},g^{-n,*}}_\mu(i^n_{T-k-1},T-k-1)\nonumber\\
&=\underset{g^n_{T-k-1}}{\sup}\ \mathbb{E}_{\mu^n}\{\mathbb{E}_{i^{-n}_{T-k-1}}^{g^{-n,*}}\{ r^n_{T-k-1}\vert i^n_{T-k-1}\}\} \nonumber\\
&- \mathbb{E}^{g_{T-k-1}^n(i^n_{T-k-1})}\{a^n_{T-k-1}\vert i^n_{T-k-1}\} c^n\nonumber\\
&=\underset{g^n_{T-k-1}}{\sup}\ E_{\mu^n}\{E_{i^{-n}_{T-k-1}}^{g^{-n,*}}\{ r^n_{T-k-1}\vert i^n_{T-k-1}\}\}\nonumber\\& - g_{T-k-1}^n(i^n_{T-k-1})(a^n_{T-k-1} =1) c^n\nonumber\\
&=E_{\mu^n}\{E_{i^{-n}_{T-k-1}}^{g^{-n,*}}\{ r^n_{T-k-1} \vert i^n_{T-k-1}\}\},
\end{align}
which is clearly attained when $g_{T-k-1}^n(i^n_{T-k-1})(a^n_{T-k-1} =1)=0$ $\forall n, i^n_{T-k-1}$. This completes the if part of the proof.

Conversely, let $g^*=g^{NC}$ and $(g^*,\mu)$ an assessment with $\mu$ consistent with $g^*$. Then, note that by Lemma \ref{Lem_reward} we have that $\forall n,t, i^n_t$ $J^{n,g^{n,NC}, g^{-n,NC}}_\mu(i^n_t, t)=0$. The expected instantaneous reward at any time $t$ given $i^n_t$, for a strategy $g^n$, given $g^{-n}=g^{-n,NC}$, is
 \begin{align}
& E_{\mu^n}[E_{i^{-n}_{t}}^{g^{-n,NC}}[ r^n_t\vert i^n_t]] - g^n(i^n_t)(a^n_t=1) c^n\nonumber\\
& = - g^n(i^n_t)(a^n_t=1) c^n\leq 0.
\end{align}
Hence, the sum of the expected instantaneous rewards from times $t$ to $T$ is a non-positive random variable and its expectation is also non-positive. This implies that $\forall n,t,i^n_t$ and for arbitrary $g^n$
\begin{equation}
J^{n,g^n,g^{-n,NC}}_\mu(i^n_t, t) \leq 0 = J^{n,g^{n,NC},g^{-n,NC}}_\mu(i^n_t, t).
\end{equation}
\section*{Proof of Lemma \ref{sigma_common}}
By definition, it is $\mathbb{P}^{\sigma^{-n}}(a^{-n}_t|i^{-n}_t)=\sigma^{-n}(s_t,\pi^{-n}_t)(a^{-n}_t)$ for every $i^{-n}_t:%l(i^{-n}_t)=
\pi^{-n}_t(X_t)=\mathbb{P}(X_t|i^{-n}_t)$. 
Moreover, if $s_t=1$. then $i^{-n,p}_t=\emptyset$ and $i^n_t=i^{-n}_t\Rightarrow\pi^n_t=\pi^{-n}_t$, while if $s_t=0$, it is $\sigma^{-n}(s_t,\pi^{-n}_t)(a^{-n}_t=0)=1$ for every $\pi^{-n}_t$. Thus, in both cases $i^n_t$ and as a result $(\pi^n_t,s_t)$ suffices to compute $\sigma^{-n}(s_t,\pi^{-n}_t)(a^{-n}_t)$.

Regarding part $2$ of the Lemma, given agent $-n$ follows a CGT strategy $\sigma^{-n,\Pi^{-n,c}}$ and for given $\pi^n_t,y^n_t,z^{-n}_t,a^n_t,a^{-n}_t$, where $a^{-n}_t$ is given by \eqref{sigma_ind1}, %\eqref{sigma_ind2}, 
agent $n$'s belief $\pi^n_{t+1}(x_{t+1})$ is updated as follows
\begin{small}
\begin{align}
\label{belief_updateaux}
&\pi^n_{t+1}(x_{t+1})=\mathbb{P}^{g^n,\sigma^{-n}}(X_{t+1}=x_{t+1}|i^n_{t+1})\nonumber\\
&=\mathbb{P}^{g^n,\sigma^{-n}}(x_{t+1}|i^n_t,y^n_t,z^{-n}_t,a^n_t,a^{-n}_t,s_{t+1})\nonumber\\
&=\frac{\mathbb{P}(s_{t+1}|s_t,a^n_t,a^{-n}_t)\mathbb{P}^{g^n,\sigma^{-n}}(x_{t+1},i^n_t,y^n_t,z^{-n}_t,a^n_t,a^{-n}_t)}{\mathbb{P}(s_{t+1}|s_t,a^n_t,a^{-n}_t)\mathbb{P}^{g^n,\sigma^{-n}}(i^n_t,y^n_t,z^{-n}_t,a^n_t,a^{-n}_t)}\nonumber\\
&=\frac{\sum_{x_t}\mathbb{P}(x_{t+1}|x_t)L^n(x_t)\mathbb{P}^{\sigma^{-n}}(a^{-n}_t|i^n_t)\mathbb{P}^{g^n}(a^n_t|i^n_t)\mathbb{P}(x_t|i^n_t)}{\sum_{x_t}L^n(x_t)\mathbb{P}^{\sigma^{-n}}(a^{-n}_t|i^n_t)\mathbb{P}^{g^n}(a^n_t|i^n_t)\mathbb{P}(x_t|i^n_t)},
\end{align}
\end{small}\noindent
where
\begin{align}
    L^n(x_t)=\mathbb{P}(y^n_t|x_t)\mathbb{P}(z^{-n}_t|x_t,a^{-n}_t).
\end{align}
$\mathbb{P}(z^{-n}_t|x_t,a^{-n}_t)$ is given by \eqref{y_z}. Due to to the first part of the Lemma, \eqref{belief_updateaux} yields
\begin{small}
\begin{align}
\label{belief_update}
    &\pi^n_{t+1}(x_{t+1})=\frac{\sum_{x_t}\mathbb{P}(x_{t+1}|x_t)L^n(x_t)\sigma^{-n}(s_t,\pi^n_t)(a^{-n}_t)\pi^n_t(x_t)}{\sum_{x_t}L^n(x_t)\sigma^{-n}(s_t,\pi^n_t)(a^{-n}_t)\pi^n_t(x_t)}\nonumber\\
    &=\frac{\sum_{x_t}\mathbb{P}(x_{t+1}|x_t)\mathbb{P}(y^n_t|x_t)\mathbb{P}(z^{-n}_t|x_t,a^{-n}_t)\pi^n_t(x_t)}{\sum_{x_t}\mathbb{P}(y^n_t|x_t)\mathbb{P}(z^{-n}_t|x_t,a^{-n}_t)\pi^n_t(x_t)}.
\end{align}
\end{small}\noindent
Note that no assumption on agent $n$'s strategy has been made.

Regarding part $3$ of the Lemma, observing the expected instantaneous reward function \eqref{ins_reward} depends on $\pi^n_t(X_t)=\mathbb{P}(X_t|i^n_t)$ and $\mathbb{P}(Z^{-n}_t|X_t,i^n_t)$. The term $\mathbb{P}(Z^{-n}_t|X_t,i^n_t)$ is a function of the other agent's strategy $g^{-n}$ as can be seen, by \eqref{ins_reward2}, \eqref{ins_reward3}, \eqref{infer_action}. However, if the other agent follows a CGT strategy then $\mathbb{P}(A^{-n}_t|X_t,i^n_t)$ is given by $\sigma^{-n}(s_t,\pi^n_t)$ from the first part of the Lemma.

Hence, if $s_t=1$, then $i^n_t=i^{-n}_t=i^c_t$ (and as a result, $\pi^n_t=\pi^{-n}_t$) and by part $1$ of Lemma \ref{Lem_reward} we have
\begin{align}
\label{GT_Reward1}
&\mathbb{E}\{r^n_t({X}_t,Y^n_t,Z^{-n}_t(Y^{-n}_t,A^{-n}_t)|I^n_t=i^n_t)\}\nonumber\\
&=\sigma^{-n}_t(i^{-n}_t)(a^{-n}_t=1)I(X_t;Y^{-n}_t|Y^n_t,i^n_t)\nonumber\\
&=\sigma^{-n}(s_t=1,\pi^n_t)(a^{-n}_t=1)I(X_t;Y^{-n}_t|Y^n_t,i^n_t)\nonumber\\
&=\mathds{1}_{\{s_t=1,\pi^n_t\in\Pi^{-n,c}\}}I(X_t;Y^{-n}_t|Y^n_t,i^n_t).
\end{align}
On the other hand, if $s_t=0$, it is $\sigma^{-n}(s_t=0,\pi^{-n}_t)(a^{-n}_t=0)=1$ for every $\pi^{-n}_t$, meaning that the equivalent behavioral strategy is $g^{-n}_t(i^{-n}_t)(a^{-n}_t=0)=1$, for every $i^{-n,p}_t$ and by part $2$ of Lemma \ref{Lem_reward} it is
\begin{align}
\label{GT_Reward2}
&\mathbb{E}\{r^n_t({X}_t,Y^n_t,Z^{-n}_t(Y^{-n}_t,A^{-n}_t)|i^n_t)\}\nonumber\\
&=g^{-n}_t(i^{-n}_t)(a^{-n}_t=1)I(X_t;Y^{-n}_t|Y^n_t,i^n_t)\nonumber\\
&=\sigma^{-n}(s_t=0,%l(i^{-n}_t)=
\pi^{-n}_t)(a^{-n}_t=1)I(X_t;Y^{-n}_t|Y^n_t,i^n_t)\nonumber\\
&=0.
\end{align}
Combining \eqref{GT_Reward1} and \eqref{GT_Reward2}, we get \eqref{GT_Reward}, \eqref{GT_Reward_aux}.
\section*{Proof of Theorem \ref{ST-POMDP}}
Define a new system state as $\tilde{X}_t=(X_t,S_t,\Pi^n_t)$ and observation $\tilde{Y}^n_t=(Y^n_{t-1},Z^{-n}_{t-1}
,A^{-n}_{t-1})$. A POMDP consists of a system state $\tilde{X}_t\in\mathcal{\tilde{X}}$, an observation process $\tilde{Y}^n_t\in\hat{\mathcal{Y}}^n$, an action process $A^n_t\in\mathcal{A}$. To show that  agent $n$'s best response problem is a POMDP problem we need the following conditions
\begin{enumerate}
\item The state dynamics are Markovian, i.e.,
\begin{align}
\label{state_obs}
\mathbb{P}^{\sigma}_f(\tilde{X}_{t+1}|\tilde{X}_{0:t},\tilde{Y}^n_{0:t},A^n_{0:t})=\mathbb{P}^{\sigma}_f(\tilde{X}_{t+1}|\tilde{X}_t,A^n_t).
\end{align}
\item The observation dynamics satisfy
\begin{align}
\label{obs_form}
\mathbb{P}^{\sigma}_f(\tilde{Y}^n_t|\tilde{X}_{0:t-1},\tilde{Y}^n_{0:t-1},A^n_{0:t-1})=
\mathbb{P}^{\sigma}_f(\tilde{Y}^n_t|\tilde{X}_{t-1},A^n_{t-1}).
\end{align}
\item Agent $n$'s instantaneous utility is a function of the information state $(S_t,\Pi^n_t)$ and action $A^n_t$.
\end{enumerate}
We note that for condition $2$, we follow the timing structure of \cite{Witsenhausen} and \cite{Nayyar_2011}, where agent's observation is a function of the previous state and action. Also,  to save notation we write $\mathbb{P}^{\sigma}_f(\cdot)$ instead of $\mathbb{P}^{\sigma^{-n}}_f(\cdot)$.

Given the fact that agent $-n$ follows the CGT strategy $\sigma^{-n,\Pi^{-n,c}}(s_t,\pi^{-n}_t)$, we have
\begin{align*}
&\mathbb{P}^{\sigma}_f(\tilde{x}_{t+1}|\tilde{x}_{0:t},\tilde{y}^n_{0:t},a^n_{0:t})\nonumber\\
&=\mathbb{P}^{\sigma}_f(x_{t+1},s_{t+1},\pi^n_{t+1}|x_{0:t},s_{0:t},\pi^n_{0:t},y^n_{0:t-1},z^{-n}_{0:t-1},a^{-n}_{0:t-1},a^n_{0:t})\nonumber\\
&=\mathbb{P}(x_{t+1}|x_t)\sum_{y^n_t,z^{-n}_t,a^{-n}_t}\mathds{1}_{\{\pi^n_{t+1}=f(\pi^n_t,y^n_t,z^{-n}_t,a^{-n}_t)\}}\nonumber\\
&\times\mathbb{P}(s_{t+1}|s_t,a^n_t,a^{-n}_t)\nonumber\\
&\times\mathbb{P}(z^{-n}_t,a^{-n}_t,y^n_t|x_{0:t},s_{0:t},\pi^n_{0:t},y^n_{0:t-1},z^{-n}_{0:t-1},a^{-n}_{0:t-1},a^n_{0:t})\nonumber\\
&=\mathbb{P}(x_{t+1}|x_t)\sum_{y^n_t,z^{-n}_t,a^{-n}_t}\mathds{1}_{\{\pi^n_{t+1}=f(\pi^n_t,y^n_t,z^{-n}_t,a^{-n}_t)\}}\mathbb{P}(y^n_t|x_t)\nonumber\\
&\times\mathbb{P}(s_{t+1}|s_t,a^n_t,a^{-n}_t)\mathbb{P}(z^{-n}_t|x_t,a^{-n}_t)\sigma^{-n}(s_t,\pi^n_t)(a^{-n}_t),
\end{align*}
where $\mathbb{P}(z^{-n}_t|x_t,a^{-n}_t)$ is given by \eqref{y_z} and $\sigma^{-n}(s_t,\pi^n_t)(a^{-n}_t)$ by \eqref{sigma_ind1}, 
\eqref{sigma_inference}. 
Hence, 
we conclude that the next state depends only on the value of the state variables from the previous time instant (i.e., $\tilde{x}_t$) and the previous action $a^n_t$ and as a result, the system dynamics are of the form of \eqref{state_obs}.

With regards to condition $2$, we have
\begin{align*}
&\mathbb{P}^{\sigma}(\tilde{y}^n_t|\tilde{x}_{0:t-1},\tilde{y}^n_{0:t-1},a^n_{0:t-1})\\
&=\mathbb{P}^{\sigma}(y^n_{t-1},z^{-n}_{t-1},a^{-n}_{t-1}|x_{0:t-1},s_{0:t-1},\pi^n_{0:t-1},\tilde{y}^n_{0:t-1},a^n_{0:t-1})\nonumber\\
&=\mathbb{P}(y^n_{t-1}|x_{t-1})\mathbb{P}(z^{-n}_{t-1}|x_{t-1},a^{-n}_{t-1})\mathbb{P}(a^{-n}_{t-1}|\pi^n_{t-1},s_{t-1})\nonumber\\
&=\mathbb{P}(y^n_{t-1}|x_{t-1})\mathbb{P}(z^{-n}_{t-1}|x_{t-1},a^{-n}_{t-1})\sigma^{-n}(s_{t-1},\pi^n_{t-1})(a^{-n}_{t-1}).\nonumber
\end{align*}
Following the same reasoning as in condition $1$, we conclude that observation $\tilde{y}^n_t$ is a function of $\tilde{x}_{t-1}$ and as a result of the form of \eqref{obs_form}. Hence, condition $2$ holds.

From part $3$ of Lemma \ref{sigma_common}, we observe that given agent $-n$ follows a CGT strategy, the expected instantaneous reward function 
is a function of $s_t,\pi^n_t,a^n_t$, meaning that it is a function of $\tilde{x}_t$ and $a^n_t$ and as a result, condition $3$ holds.

Moreover, $(S_t,\Pi^n_t)$ is an {\em information state}. In order to establish that, we have to show that $(1)$ it can be updated recursively, i.e., $\{S_{t+1},\Pi^n_{t+1}\}$ can be updated by the previous $\{S_t,\Pi^n_t\}$ and the newly acquired information $\tilde{Y}^n_{t+1},A^n_t$, $(2)$ that agent $n$'s belief on $\{S_{t+1},\Pi^n_{t+1}\}$ conditioned on $\{S_t,\Pi^n_t,A^n_t\}$ is independent of the whole history $I^n_t$ and $(3)$ it is sufficient to evaluate agent $n$'s expected utility for every action $a^n_t$.

Given that agent $-n$ follows CGT strategy, \eqref{s_t-evolution} yields
\begin{align}
s_{t+1}=\mathds{1}_{\{s_t=a^n_t=\sigma^{-n}_t(s_t,\pi^n_t)=1\}}.
\end{align}
Thus, $S_{t+1}$  is updated recursively as a function of $s_t,\pi^n_t$ and $a^n_t$. The same is true for $\pi^n_{t+1}$ ({\em see} \eqref{belief_update}), as it is $\pi^n_{t+1}(x_{t+1})=f(\pi^n_t,y^n_t,z^{-n}_t,a^{-n}_t)$, meaning that $\pi^n_{t+1}$ is recursively updated by the previous $\tilde{x}_t$ and the new information $\tilde{y}^n_{t+1}$. Thus, condition $(1)$ holds.

Regarding condition $(2)$, we have
\begin{align}
&\mathbb{P}^{\sigma}_f(s_{t+1},\pi^n_{t+1}|i^n_t,a^n_t,s_t,\pi^n_t)=\sum_{a^{-n}_t}\mathbb{P}(s_{t+1}|s_t,a^n_t,a^{-n})\nonumber\\
&\times\mathbb{P}^{\sigma}_f(\pi^n_{t+1}|i^n_t,a^n_t,a^{-n}_t,s_t,\pi^n_t)\mathbb{P}^\sigma(a^{-n}_t|i^n_t,a^n_t,s_t,\pi^n_t)
\nonumber\\
&=\sum_{a^{-n}_t}\mathbb{P}(s_{t+1}|s_t,a^n_t,a^{-n})\sum_{y^n_t,z^{-n}_t}\mathds{1}_{\{\pi^n_{t+1}=f(\pi^n_t,y^n_t,z^{-n}_t,a^{-n}_t)\}}\nonumber\\
&\times\mathbb{P}(y^n_t,z^{-n}_t|i^n_t,s_{0:t},a^n_t,a^{-n}_t)\sigma(s_t,\pi^n_t)(a^{-n}_t)\nonumber\\
&=\sum_{a^{-n}_t}\mathbb{P}(s_{t+1}|s_t,a^n_t,a^{-n})\sum_{y^n_t,z^{-n}_t}\mathds{1}_{\{\pi^n_{t+1}=f(\pi^n_t,y^n_t,z^{-n}_t,a^{-n}_t)\}}\nonumber\\
&\times\sum_{x_t}\mathbb{P}(y^n_t|x_t)\mathbb{P}(z^{-n}_t|x_t,a^{-n}_t)\pi^n_t(x_t)\sigma(s_t,\pi^n_t)(a^{-n}_t),
\end{align}
which depends on $s_t,\pi^n_t,a^n_t$ and it is independent of $i^n_t$. Condition $(3)$ is true from part $3$ of Lemma \ref{sigma_common}. 
\section*{Proof of Theorem \ref{Closeness}}
From Theorem \ref{ST-POMDP}, we showed that given that agent $-n$ follows a CGT strategy, agent $n$ without loss of optimality, can condition her strategy on $(S,\Pi^n)$. We now check whether a CGT strategy for agent $n$ is {\em sequentially rational}.
\begin{enumerate}
\item 
If $s=0$, then agent $n$'s optimal action is $a^n=0$, meaning that agent $n$ does not have any benefit from deviating from CGT, for all $\pi^n$. This is because
\begin{align*}
&\hspace{-5mm}Q^n(s=0,\pi^n,a^n=0)\hspace{-1mm}\geq \hspace{-1mm}Q^n(s=0,\pi^n,a^n=1)\hspace{-1mm}\overset{\eqref{Q_1}}\Leftrightarrow\hspace{-1mm}0\hspace{-1mm}\geq\hspace{-1mm}-c^n.
\end{align*}
\item 
If $s=1$, then agent $n$ selects $a^n=1$, if the following holds
\begin{align}
\label{BE_inequality2}
&\hspace{-7mm}Q^n(s=1,\pi^n,a^n=1)\geq Q^n(s=1,\pi^n,a^n=0)\nonumber\\
&\hspace{-7mm}\Leftrightarrow -c^n+\delta\mathbb{E}^{\sigma^{-n}}\{V^n(s',\pi'^n)|\pi^n,s=1,a^n=1\}\geq 0,
\end{align}
while agent $n$ selects $a^n=0$, if the following holds
\begin{align*}
&\hspace{-5mm}Q^n(s=1,\pi^n,a^n=1)<Q^n(s=1,\pi^n,a^n=0)\nonumber\\
&\hspace{-5mm}\Leftrightarrow-c^n+\delta\mathbb{E}^{\sigma^{-n}}\{V^n(s',\pi'^n)|\pi^n,s=1,a^n=1\}<0.
\end{align*}
Inequality \eqref{BE_inequality2} defines these $\pi^n\in\Delta(\mathcal{X})$ that comprise a region $C\subseteq \Delta(\mathcal{X})$ such that the strategy $\sigma^{n,C}(s_t,\pi^n_t)$ is optimal for agent $n$. Thus, following a CGT strategy is sequentially rational for agent $n$, given that agent $-n$ follows a CGT strategy.
\end{enumerate}
\section*{Proof of Theorem \ref{pbequiv}}
In Theorem \ref{ST-POMDP}, we showed that $(S_t,\Pi^n_t)$ is an {\em information state}, which is updated recursively and, under a CGT strategy profile $\sigma$, $s_{j+1},\pi^n_{j+1}$ is independent of other agents' private information given $s_j,\pi^n_j$ $\forall j, n$. Thus, given $\pi^{n,X}(\mu)=\pi^{n,X}(\mu')$, we have $\forall t$, \begin{align}
& E^{\sigma}_{\mu^n_t}\{\sum\limits_{j=t}^\infty \delta^j \tilde{R}^n(S_j,\Pi^n_j,a^n_j)\vert i^n_t\}\nonumber\\
&=E^{\sigma}\{\sum\limits_{j=t}^\infty \delta^j \tilde{R}^n(S_j,\Pi^n_j,a^n_j)\vert \pi^{n,\mu}_t(i^n_t), s_t(i^n_t)\}\nonumber\\
&=E^{\sigma}\{\sum\limits_{j=t}^\infty \delta^j \tilde{R}^n(S_j,\Pi^n_j,a^n_j)\vert \pi^{n,\mu'}_t(i^n_t), s_t(i^n_t)\}\nonumber\\
& =E^{\sigma}_{\mu'^n}\{\sum\limits_{j=t}^\infty \delta^j \tilde{R}^n(S_j,\Pi^n_j,a^n_j)\vert i^n_t\}.
\end{align}
Thus, equality of PBE values under $\mu,\mu'$ is evident provided that $\pi^{n,X}(\mu)=\pi^{n,X}(\mu')$.

Now for sequential rationality, suppose $(\sigma^*,\mu)$ is a PBE with $\sigma^*$ a CGT profile and that under $\mu'$, $\sigma^*$ is suboptimal for $n$. Then by Theorem \ref{Closeness}, there is an optimal CGT response $\sigma'^{n,*}\neq \sigma^{n,*}$ to $\sigma^{-n,*}$. Thus, from the first part of the Theorem,
\begin{align}
&E^{\sigma^*}_{\mu^n_t}[\sum\limits_{j=t}^\infty \delta^j \tilde{R}^n_j(S_j,\Pi^n_j,a^n_j)\vert i^n_t]=\nonumber \\
& E^{\sigma^*}_{\mu'^n_t}[\sum\limits_{j=t}^\infty \delta^j \tilde{R}^n_j(S_j,\Pi^n_j,a^n_j)\vert i^n_t]<\nonumber \\
& E^{\sigma'^{n,*},\sigma^{-n,*}}_{\mu'^n_t}[\sum\limits_{j=t}^\infty \delta^j \tilde{R}^n_j(S_j,\Pi^n_j,a^n_j)\vert i^n_t] \overset{(a)}= \nonumber\\
&  E^{\sigma'^{n,*},\sigma^{-n,*}}_{\mu^n_t}[\sum\limits_{j=t}^\infty \delta^j \tilde{R}^n_j(S_j,\Pi^n_j,a^n_j)\vert i^n_t] \Rightarrow \bot,
\end{align}
since $\sigma^*$ was sequentially rational with respect to $\mu$. (a) is true because $(\sigma'^{n,*},\sigma^{-n,*})$ is a CGT profile.
\section*{Proof of Proposition \ref{Prp_regions}}
It is enough to show that for a $\pi\in\Delta(\mathcal{X})$, such that $\pi\notin\ C\Rightarrow \pi\notin O^n(C)$. Let $\pi\notin C$ and suppose $\pi \in O^n(C)$. Then from the definition of $O^n(\cdot)$, we have that $Q^n(s=1,\pi,1)\geq Q^n(s=1,\pi,0)$. Moreover, given that $\sigma^{-n}(s=1,\pi)=\mathds{1}_{\{\pi \in C,s=1\}} = 0$, since $\pi\notin C$, and as a result it is $\tilde{r}(s,\pi)=0$ and $s'=0$. Hence,
\begin{align}
&Q^n(s=1,\pi,a^n=1)\geq	Q^n(s=1,\pi,a^n=0)\nonumber\\
&-c^n+\delta\mathbb{E}^{\sigma^{-n}}\{V^n(s'=0,f(\pi,y^n,z^{-n}=\epsilon))|\pi,s=1\}\geq\nonumber\\
&\delta\mathbb{E}^{\sigma^{-n}}\{V^n(s'=0,f(\pi,y^n,z^{-n}=\epsilon))|\pi,s=1\}\nonumber\\
&\Leftrightarrow
-c^n\geq 0 \Rightarrow \bot.
\end{align}
Therefore, $\pi \notin O^n(C)$.

For the second part of the Proposition, note that from the first part we have
\begin{align}
&\Pi^{1,c}=O^1({\Pi^{2,c}})\subseteq\Pi^{2,c},\\
&\Pi^{2,c}=O^2({\Pi^{1,c}})\subseteq\Pi^{1,c}.
\end{align}
Thus, it is $\Pi^{1,c,*}=\Pi^{2,c,*}$.
\section*{Proof of Proposition \ref{CGT_PBE}}
In order to check whether a pair of CGT strategies are sequentially rational, beliefs on past states trajectory and other agent's private information are irrelevant ({\em see} Remark \ref{rem_spe}). 
To check whether a pair of strategies $\sigma^*=(\sigma^{1,\Pi^c},\sigma^{2,\Pi^c})$ forms an SPE, we need to check that $\forall t, i^n_t, \sigma^n$, 
\begin{align}
& E^{\sigma^{n,\Pi^c},\sigma^{-n,\Pi^c}}_{\mu^n_t}\{\sum\limits_{j=t}^\infty \delta^j \tilde{R}^n_j(S_j,\Pi^n_j,a^n_j)\vert i^n_t\}\nonumber\\
&=E^{\sigma^{n,\Pi^c},\sigma^{-n,\Pi^c}}\{\sum\limits_{j=t}^\infty \delta^j \tilde{R}^n_j(S_j,\Pi^n_j,a^n_j)\vert\pi^n_t,s_t\} \\
&\geq E^{\sigma^{n},\sigma^{-n,\Pi^c}}\{\sum\limits_{j=t}^\infty \delta^j \tilde{R}^n_j(S_j,\Pi^n_j,a^n_j)\vert\pi^n_t,s_t\}.
\end{align}
This is true because (83) is equal to $V^{n,\sigma^n,\sigma^{-n,\Pi^c}}(s_t,\pi^n_t)$ and (82) is $V^{n,*,\sigma^{-n,\Pi^c}}(s_t,\pi^n_t)$ by definition of an equilibrium region, since $\Pi^c\in\mathcal{E}$. 
\section*{Proof of Theorem \ref{Value_function_increasing_c}}
For the first part, if $s=0$, we have $V^{n,*,C}(s=0,\pi^n)=V^{n,*,C'}(s=0,\pi^n)=0$, $\forall \pi$. As such, we need to consider only the case where $s=1$. Let $C^*=O^n(C)$ and note that by part $1)$ of Proposition \ref{Prp_regions}, $C^*\subseteq C \subseteq C'$. \\
\\
Under a fixed strategy profile consisting of CGT strategies, by Lemma  \ref{sigma_common}, the agents' actions are a function of agent $n$'s information state. As such, when computing $V^{n,A,B}$ the expectation is over all possible trajectories of information states $\mathcal{D}$. Define, for $X\subseteq \Delta(\mathcal{X})$ the event $T_k(X)=\{\forall t<k, \pi_t\in X \}\cap \{\pi_k \notin X\}$. Note that the sets $T_k(X)$ for any fixed $X$ and for $k\geq 0$ form a partition of $\mathcal{D}$. Therefore, we have
{\small{
\begin{align}
&V^{n,C^*,C} (s=1, \pi) = \mathbb{E}^{C^*,C}\{ \sum\limits_{i=0}^\infty \delta^i \tilde{R}^n(s_i,a^n_i,\pi_i) \vert \pi,s=1\}\nonumber\\
&\overset{(a)}{=}\sum\limits_{k=0}^\infty \mathbb{E}^{C^*,C}\{\sum\limits_{i=0}^\infty \delta^i \tilde{R}^n(s_i,a^n_i,\pi_i) \vert \pi,s=1, T_k(C^*)\}\nonumber\\
&\times\mathbb{P}(T_k(C^*)\vert \pi,s=1)\nonumber\\
&\overset{(b)}{=}\sum\limits_{k=0}^\infty \mathbb{E}^{C^*,C}\{\sum\limits_{i=0}^k \delta^i \tilde{R}^n(s_i,a^n_i,\pi_i) \vert \pi,s=1, T_k(C^*) \}\nonumber\\
&\times\mathbb{P}(T_k(C^*)\vert\pi,s=1)
\nonumber\\
&\overset{(c)}{\leq}\sum\limits_{k=0}^\infty \mathbb{E}^{C^*,C'}\{\sum\limits_{i=0}^k \delta^i \tilde{R}^n(s_i,a^n_i,\pi_i) \vert \pi,s=1, T_k(C^*) \}\nonumber\\
&\times\mathbb{P}(T_k(C^*)\vert\pi,s=1)
\nonumber\\
&\overset{(d)}{=}\sum\limits_{k=0}^\infty \mathbb{E}^{C^*,C'}\{\sum\limits_{i=0}^\infty \delta^i \tilde{R}^n(s_i,a^n_i,\pi_i) \vert \pi,s=1, T_k(C^*) \}\nonumber\\
&\times\mathbb{P}(T_k(C^*)\vert\pi,s=1)
\nonumber\\
&= V^{n,C^*,C'}(s=1,\pi)\leq V^{n,*,C'}(s=1,\pi).
\end{align}}}
\normalsize
(a) is true by the law of total expectations and (b) is true because  under the strategy profile $(\sigma^{n,C^*},\sigma^{-n,C})$, given $T_k(C^*)$, $\sum\limits_{i=k+1}^\infty \delta^i \tilde{R}^n(s_i,a^n_i,a^{-n}_i,\pi_i) =0$ for any trajectory (since agent $n$ deviates from cooperation at time $k$ and thus $s_i=0,\forall i\geq k+1$). The argument for (c) is as follows. Because only the first $k$ steps of each trajectory contribute to the conditional expectation given $T_k(C^*)$, it is enough to consider the epectation over trajectories truncated at length $k$. Evidently all such trajectories are equiprobable under $(\sigma^{n,C^*},\sigma^{n,C})$ and $(\sigma^{n,C^*},\sigma^{n,C'})$ given $T_k(C^*)$, since exactly the same actions (cooperation) are taken until (and including) step $k-1$ by both agents. For the same reason, under $(\sigma^{n,C^*},\sigma^{n,C})$ and $(\sigma^{n,C^*},\sigma^{n,C'})$, $\mathbb{P}(T_k(C^*))$ is unchanged.  Now, the reward at time $k$ is non-negative and the trajectories for which the reward at time $k$ is positive under $(\sigma^{n,C^*},\sigma^{n,C})$ are also trajectories where the reward is positive under $(\sigma^{n,C^*},\sigma^{n,C'})$, since $C\subseteq C'$. Finally, (d) is true because given $T_k(C^*)$, and the strategy pair $C^*,C'$, $\sum\limits_{i=k+1}^\infty \delta^i \tilde{R}(s_i,a^n_i,\pi_i) = 0$.

 For the second part of the Theorem, by the first part we have that $V^{n,*,C'}(s=1,\pi)-V^{n,*,C}(s=1,\pi)$ is a non-negative random variable and hence its expectation is non-negative. Note that the distribution of $\pi'$ is the same under $\sigma^{-n,C}$ and $\sigma^{-n,C'}$ given $a^n=1$, $\pi \in C$ and $s=1$.
\section*{Proof of Lemma \ref{Lem_7}}
Since $C$ is an equilibrium region we have $O^n(C)=C$. Hence, from definition of $O^n(C)$, $\forall \pi \in C$ we have that 
\begin{align}
&Q^{n,*,C}(s=1,\pi,a^n=1) \geq Q^{n,*,C}(s=1,\pi,a^n=0)\nonumber\\
&\Leftrightarrow\tilde{r}(\pi,s=1)-c + \delta E\{V^{n,*,C}(s'=1,\pi')\vert \pi, s=1,a^n=1\} \nonumber\\
&\geq\tilde{r}(\pi,s=1)+ \delta E\{V^{n,*,C}(s'=0,\pi')\vert \pi, s=1,a^n=0\}\nonumber\\
&\Leftrightarrow \delta E\{V^{n,*,C}(s'=1,\pi')\vert \pi, s=1,a^n=1\} - c^n \geq 0.
\end{align}
Note that $E\{V^{n,*,C}(s'=0,\pi')\vert \pi, s=1, a^n=0\}= 0$.
\section*{Proof of Theorem \ref{Th3}}
We first prove an auxiliary Lemma.
\begin{Lem}
\label{Lem_8}
Let $\Pi_1,\Pi_2\in \mathcal{E}$. Then $\Pi_1 \cup \Pi_2 \in \mathcal{E}$.
\end{Lem}
\begin{proof}
We have from part $1)$ of Proposition \ref{Prp_regions} that $O^n(\Pi_1\cup\Pi_2)\subseteq \Pi_1\cup\Pi_2$ for $n=\{1,2\}$. We first show that $\Pi_1\subseteq O^n(\Pi_1\cup\Pi_2)$. Suppose for a contradiction that $\pi \in\Pi_1$ and $\pi \notin O^n(\Pi_1\cup\Pi_2)$. Then, 
\begin{align}
&Q^{n,*,\Pi_1\cup\Pi_2}(s=1,\pi,a^n=1) < Q^{n,*,\Pi_1\cup\Pi_2}(s=1,\pi,a^n=0)%\nonumber\\
\nonumber\\
&\Rightarrow\tilde{r}^n(\pi,s=1)-c^n+ \delta E\{  \nonumber\\
&V^{n,*,\Pi_1\cup\Pi_2}(s'=1,\pi')\vert \pi, s=1,a^n=1\}< \tilde{r}^n(\pi,s=1) \Rightarrow \nonumber\\
&\delta E\{V^{n,*,\Pi_1\cup\Pi_2}(s'=1,\pi')\vert \pi, s=1,a^n=1\} - c^n < 0.
\end{align}
Because $\Pi_1\subseteq \Pi_1\cup\Pi_2$, from the second part of Theorem \ref{Value_function_increasing_c} we have%\ref{Lem_7} and (86) gives
\begin{align*}
\delta E\{V^{n,*,\Pi_1}(s'=1,\pi')\vert \pi, s=1,a^n=1\} - c^n\leq\nonumber\\ \delta E\{V^{n,*,\Pi_1\cup\Pi_2}(s'=1,\pi')\vert \pi, s=1,a^n=1\} - c^n <0,
\end{align*}
but this is a contradiction because $\Pi_1$ is an equilibrium region and Lemma \ref{Lem_7} applies.
\end{proof}
For the first part, Zorn's Lemma guarantees at least one maximal equilibrium region. Now suppose for a contradiction that there exists two distinct maximal equilibrium regions $\Pi_1,\Pi_2$. From Lemma \ref{Lem_8} %\ref{Lem_7} 
we have $\Pi_1\cup\Pi_2\in \mathcal{E}$. Now $\Pi_1\subseteq \Pi_1\cup \Pi_2$ and thus by the maximality of $\Pi_1$ we have $\Pi_1=\Pi_1\cup\Pi_2$. Hence $\Pi_2 \subseteq \Pi_1$ and by maximality of $\Pi_2$, $\Pi_1=\Pi_2 \Rightarrow \bot$. The second part is true due to Theorem \ref{Value_function_increasing_c}.
\section*{Proof of Proposition \ref{Lem_9}}
For the first part of the Theorem, let $\pi\notin O^n(C)$ and suppose $\pi\in \Pi^*$ for a contradiction. We have,
\begin{align}
&Q^{n,*,C}(s=1,\pi,a^n=1)<Q^{n,*,C}(s=1,\pi,a^n=0)\nonumber\\
&\Rightarrow-c^n + \delta E\{V^{n,*,C}(s'=1,\pi')\vert s,\pi,a^n=1 \} < 0.
\end{align}
Then from \eqref{BE_increasing}, since $\Pi^*\subseteq C$ we have
\begin{align}
-c^n + \delta E\{V^{n,*,\Pi^*}(s'=1,\pi')\vert s,\pi,a^n=1 \} < 0,
\end{align}
and since $\Pi^*\in\mathcal{E}$, this contradicts Lemma \ref{Lem_7}.
\section*{Proof of Theorem \ref{bound_region}}
Let $c^n = \tilde{\epsilon}
    r^{C}_{inf}$
and $\delta\geq\tilde{\epsilon} >0$, which is possible since $r^{C}_{inf} > 0$ as $C$ is positive absorbing. For this choice we have $r^{C}_{inf} \geq \frac{c^n}{\delta}$. Now, suppose agent $-n$ cooperates in $C$. Then, agent $n$ will cooperate at $\pi \in C$ if the following is true.
\begin{align}
\label{opt_cooperation}
    &Q^{n,C}(s=1, \pi, a^n = 1)\geq Q^{n,C}(s=1, \pi, a^n = 0)\Leftrightarrow\nonumber\\
    &E\{V^{n,*,C}(s'=1,\pi')\vert \pi, s=1,a^n=1\}\geq \frac{c^n}{\delta}.
\end{align}
Note that 
\begin{align}
\label{inequalitys}
&E\{V^{n,*,C}(s'=1,\pi')\vert \pi, s=1,a^n=1\}\nonumber\\
&\overset{(a)}\geq\mathbb{E}\{\sum^{\infty}_{t=0}\delta^t(\tilde{r}^n(s_t=1,\pi_t)-c^n)\}\nonumber\\
&\overset{(b)}\geq\sum^{\infty}_{t=0}\delta^t(r^C_{inf}-c^n)=\frac{r^C_{inf}-c^n}{1-\delta},
\end{align}
where $(a)$ is true because the optimal continuation value $E\{V^{n,*,C}(s'=1,\pi')\vert \pi, s=1,a^n=1\}$ is greater or equal than the always cooperate strategy. The expected continuation value of the always cooperate strategy is given by $\mathbb{E}\{\sum^{\infty}_{t=0}\delta^t(\tilde{r}^n(s_t=1,\pi_t)-c^n)\}$ because $C$ is absorbing. Finally, $(b)$ holds because the expected accumulated rewards are lower bounded by the infimum of the rewards $r^{C}_{inf}-c^n$ at every time step. 
Now note, that
$$
r^{C}_{inf} \geq \frac{c^n}{\delta},
$$
if and only if 
$$
\frac{r^{C}_{inf} - c^n}{1-\delta} \geq \frac{c^n}{\delta}.
$$
Thus, \eqref{opt_cooperation} is satisfied due to our choice of $c^n$ and hence, agent $n$ cooperates in $C$. As $n$ was arbitrary, this proves that $C$ is an equilibrium region for this value of the transmission cost.

Now, for the second part of the theorem, note that while $-n$ cooperates, the following is true
\begin{align}
\label{b_u}
    &{\pi^{n}}^{'}(x')=f(\pi^n,y^n,y^{-n},a^{-n}=1)(x')\nonumber\\
    &=\frac{\sum_{x}\mathbb{P}(x'|x)\mathbb{P}(y^n|x)\mathbb{P}(y^{-n}|x)\pi^n(x)}
{\sum_{x}\mathbb{P}(y^n|x)\mathbb{P}(y^{-n}|x)\pi^n(x)},
\end{align}
for any $x'\in\mathcal{X}$. Then, it is easy to verify that the following is true for every $x'$.
\begin{align}
    \lambda_{min}(x') \leq \pi^{n'}(x') \leq \lambda_{max}(x'),
\end{align}
since for any $\pi^n, y^n, y^{-n}$, 
\begin{align}
& \frac{\sum_{x}\mathbb{P}(x'|x)\mathbb{P}(y^n|x)\mathbb{P}(y^{-n}|x)\pi^n(x)}
{\sum_{x}\mathbb{P}(y^n|x)\mathbb{P}(y^{-n}|x)\pi^n(x)}
\geq\nonumber \\
& \frac{\sum_{x} \lambda_{min}(x') \mathbb{P}(y^n|x)\mathbb{P}(y^{-n}|x)\pi^n(x)}
{\sum_{x}\mathbb{P}(y^n|x)\mathbb{P}(y^{-n}|x)\pi^n(x)}= %\nonumber\\
\lambda_{min}(x'),
\end{align} 
and similarly for $\lambda_{max}$.

The above relation implies $C= \Lambda \cap \Delta(\mathcal{X})$ is absorbing.  

Moreover, if $C$ is positive absorbing, then by choosing $c^n,$ as in part $1$ of the Theorem ensures that $\Delta(\mathcal{X})$ is an equilibrium region (and hence $\Pi^* = \Delta(\mathcal{X})$). This is true due to the following. Given that agent $-n$ cooperates in $\Delta(\mathcal{X})$, then agent $n$ cooperates in $\Delta(\mathcal{X})$ if and only if $\forall \pi\in\Delta(\mathcal{X})$,
\begin{align}
&Q^{n,\Delta(\mathcal{X})}(s=1, \pi, a^n = 1)\geq Q^{n,\Delta(\mathcal{X})}(s=1, \pi, a^n = 0)\Leftrightarrow\nonumber\\
&E\{V^{n,*,\Delta(\mathcal{X})}(s'=1,\pi')\vert \pi^n, s=1,a^n=1\}\geq \frac{c^n}{\delta}.
\end{align}
Note that for any $\pi_0 \in \Delta(\mathcal{X})$, if $s_t=1$, then $\pi_t$ (given by \eqref{b_u}) for all $t>0$ lies in $C$. Hence, since $C$ is reached with probability $1$ in one step (i.e. $\pi'\in C$) and $C$ is absorbing,
\begin{align}
    &E\{V^{n,*,\Delta(\mathcal{X})}(s'=1,\pi')\vert \pi, s=1,a^n=1\} =\nonumber\\
    &E\{V^{n,*,C}(s'=1,\pi')\vert \pi, s=1,a^n=1\}.
\end{align}
Because \eqref{inequalitys} holds, the same reasoning as above yields that agent $n$ cooperates in $ \Delta(\mathcal{X})$.



%\appendices
%% you can choose not to have a title for an appendix
%% if you want by leaving the argument blank
%\section{}

%\section*{Acknowledgment}

\ifCLASSOPTIONcaptionsoff
  \newpage
\fi



%\bibliographystyle{IEEEtran}

\begin{thebibliography}{99}
%Information-sharing studies
\bibitem{Yu_2015}
C.-K. Yu, M. van~der Schaar, and A.~H. Sayed, ``Information-sharing over
  adaptive networks with self-interested agents,'' \emph{IEEE Trans.
  Signal and Information Processing over Networks}, vol.~1, no.~1, pp. 2--19, 2015.
\bibitem{Miehling_2018}
E. Miehling, M. Rasouli, and D. Teneketzis, ``A POMDP Approach to the Dynamic Defense of Large-Scale Cyber Networks," \emph{IEEE Trans. Information Forensics and Security}, vol. 13, no. 10, pp. 2490-2505, 2018.
\bibitem{Jiang_2014}
C. Jiang, Y. Chen, and K. J. R. Liu, ``Graphical evolutionary game for information diffusion over social networks,” \emph{IEEE Journal of Selected Topics in Signal Processing}, vol. 8, no. 4, pp. 524–536, 2014.
\bibitem{Chen_2010}
W. Chen, C. Wang, and Y. Wang, ``Scalable influence maximization for prevalent viral marketing in large-scale social networks,” in \emph{Proc. Int. Conf. Knowl. Discov. Data Min.}, 2010, pp. 1029–1038.
\bibitem{Tirole}
D. Fudenberg, J. and J. Tirole, \emph{Game theory}, Cambridge, Massachusetts 393, 1991.
\bibitem{Nayyar_2014}
A. Nayyar, A. Gupta, C. Langbort, and T. Başar, ``Common information based Markov perfect equilibria for stochastic games with asymmetric information: Finite games," \emph{IEEE Trans. Autom. Control} vol. 59, no. 3, pp. 555-570, 2014.
\bibitem{Vasal_2019}
D. Vasal, A. Sinha, and Achilleas Anastasopoulos, ``A systematic process for evaluating structured perfect Bayesian equilibria in dynamic games with asymmetric information," \emph{IEEE Trans. Autom. Control}, vol. 64, no. 1, pp. 78-93, 2019.
\bibitem{Chakravorty_2016}
J. Chakravorty, and A. Mahajan, ``Structural results for two-user interactive communication," in \emph{Proc. IEEE International Symposium on Information Theory (ISIT)}, July 2016, pp. 145-149.
\bibitem{Ntemos_2018}
K. Ntemos, J. Plata-Chaves, N. Kolokotronis, N. Kalouptsidis, and M. Moonen, ``Secure information sharing in adversarial adaptive diffusion networks,"  \emph{IEEE Trans. Signal and Information Processing over Networks,} vol. 4, no. 1, pp. 111-124, 2018.
\bibitem{Shoham_2008}
Y. Shoham, and K. Leyton-Brown, \emph{Multiagent Systems: Algorithmic, Game Theoretic and Logical Foundations}, Cambridge University Press, 2008.
\bibitem{Heegard_2016}
P. E. Heegaard, G. Biczok, and L. Toka, ``Sharing is power: Incentives for information exchange in multi-operator service delivery," in \emph{Proc. IEEE Global Communications Conference (GLOBECOM)}, December 2016, pp. 1-7.
\bibitem{Gal-or_2005}
E. Gal-Or, and A. Ghose, ``The economic incentives for sharing security information," \emph{Information Systems Research},  vol. 16, no. 2, pp. 186-208, 2005.
\bibitem{Naghizadeh_2017}
P. Naghizadeh, and M. Liu, ``Using Private and Public Assessments in Security Information Sharing Agreements," \emph{IEEE Trans. Information Forensics and Security}, 2019.
\bibitem{Laube_2017}
S. Laube, and R. Böhme, ``Strategic aspects of cyber risk information sharing," \emph{ACM Computing Surveys (CSUR)}, vol. 50, no. 5, pp. 1-36, 2017.
\bibitem{Bellman_1966}
R. Bellman, ``Dynamic programming," \emph{Science}, vol. 153, no. 3731, pp. 34-37, 1966.
\bibitem{Blackwell_1964}
D. Blackwell, ``Discounted dynamic programming," \emph{The Annals of Mathematical Statistics}, vol. 36, no. 1, pp. 226-235, 1965.
\bibitem{Astrom_1965}
K. J. Åström, ``Optimal control of Markov processes with incomplete state information," \emph{Journal of Mathematical Analysis and Applications}, vol. 10, no. 1, pp. 174-205, 1965.
\bibitem{Sondik_1973}
R. D. Smallwood, and E. J. Sondik, ``The optimal control of partially observable Markov processes over a finite horizon,” \emph{Operations Research}, vol. 21, no. 5, pp. 1071–1088, 1973.
\bibitem{Cassandra_thesis}
A. R. Cassandra, ``Exact and approximate algorithms for partially observable Markov decision problems,"
Ph.D. dissertation, Dept. of Computer Science, Brown University, Providence, RI, 1998.
\bibitem{Silver_2010}
D. Silver, and J. Veness, ``Monte-Carlo planning in large POMDPs," \emph{Advances in neural information processing systems}, pp.  2164-2172, 2010.
\bibitem{Bikhchandani_1992}
S. Bikhchandani, D. Hirshleifer, and I. Welch, ``A theory of fads, fashion, custom, and cultural change as informational cascades,” \emph{Journal of Political Economy}, vol. 100, no. 5, pp. 992–1026, 1992.
\bibitem{Fudenberg_1986}
E. Maskin, and D. Fudenberg, ``The folk theorem in repeated games with discounting or with incomplete information," {\em Econometrica} vol. 53, no.3, 1986.
\bibitem{Dutta_1995}
D. P. Dutta, ``A folk theorem for stochastic games," \emph{Journal of Economic Theory}, vol. 66, no. 1, pp. 1-32, 1995.
\bibitem{Horner_2011}
J. Hörner, T. Sugaya, S. Takahashi, and N. Vieille, ``Recursive methods in discounted stochastic games: An algorithm for $\delta\rightarrow1$ and a folk theorem," \emph{Econometrica}, vol. 79, no. 4, pp. 1277-1318, 2011.
\bibitem{Escobar_2013}
J. F. Escobar, and J. Toikka, ``Efficiency in games with Markovian private information," \emph{Econometrica} vol.81, no. 5, pp. 1887-1934, 2013.
\bibitem{Sugaya_2012}
T. Sugaya, ``Folk theorem in stochastic games with private state and private monitoring," Working Paper, 2012.
\bibitem{Ho_1980}
Y. C. Ho, ``Team decision theory and information structures," \emph{in Proceedings of the IEEE}, vol. 68, no. 6, pp. 644-654, 1980.
\bibitem{Nayyar_thesis}
A. Nayyar, ``Sequential Decision-Making in Decentralized systems," Ph.D. dissertation,  Univ. of Michigan, 2011.
\bibitem{Nayyar_2011}
A. Nayyar, A. Mahajan, and D. Teneketzis, ``Optimal control strategies in delayed sharing information structures," \emph{IEEE Trans. Autom. Control} vol. 56, no.7, pp. 1606-1620, 2011.
\bibitem{Vasal_2016}
D. Vasal, A. Anastasopoulos, ``Decentralized Bayesian learning in dynamic games," in Proc. \emph{54th Annual Allerton Conference on Communication, Control, and Computing (Allerton)}, Septempber 2016, pp. 264-273.
\bibitem{Nayyar_2013}
A. Nayyar, A. Mahajan, and D. Teneketzis, ``Decentralized stochastic control with partial history sharing: A common information approach," \emph{IEEE Trans. Autom. Control}, vol. 58, no. 7, pp. 1644-1658, 2013.
\bibitem{Gupta_2014}
A. Gupta, A. Nayyar, C. Langbort, and T. Basar, ``Common information based Markov perfect equilibria for linear-gaussian games with asymmetric information,” \emph{SIAM J. Control Optim.}, vol. 52, no. 5, pp. 3228–3260, 2014.
\bibitem{Ouyang_2015}
Y. Ouyang, H. Tavafoghi, and D. Teneketzis, ``Dynamic oligopoly games with private Markovian dynamics," in \emph{Proc. IEEE Decision and Control (CDC)}, December 2015, pp. 5851-5858.
\bibitem{Ouyang_2017}
Y. Ouyang, H. Tavafoghi, and D. Teneketzis, ``Dynamic games with asymmetric information: Common information based perfect bayesian equilibria and sequential decomposition," \emph{IEEE Trans. Autom. Control}, vol. 62, no. 1, pp. 222-237, 2017.
\bibitem{Tavafoghi_2016}
H. Tavafoghi, Y. Ouyang, and D. Teneketzis, ``On stochastic dynamic games with delayed sharing information structure," in \emph{Proc. IEEE Decision and Control (CDC)}, December 2016, pp. 7002-7009.
\bibitem{Vasal_2016b}
D. Vasal, and A. Anastasopoulos, ''Signaling equilibria for dynamic LQG games with asymmetric information," in \emph{Proc. IEEE Decision and Control (CDC)}, December 2016, pp. 6901-6908.
\bibitem{Gallager}
R. G. Gallager, \emph{Information theory and reliable communication}, New York: Wiley, 1968.
\bibitem{Tavafoghi_thesis}
H. Tavafoghi, ``On Analysis and Design of Cyber-physical Systems with Strategic Agents", Ph.D. dissertation, Univ. of Michigan, September 2017.
\bibitem{Ouyang-thesis}
Y. Ouyang, ``On the interaction of information and decisions in dynamic networked systems," Ph.D. Thesis, University of Michigan, 2016.
%POMDPs
\bibitem{Mailath_2006}
G. J. Mailath, and L. Samuelson, \emph{Repeated games and reputations: long-run relationships}, Oxford university press, 2006.
\bibitem{Naghshvar_2012}
M. Naghshvar, and T. Javidi, ``Extrinsic Jensen-Shannon divergence with application in active hypothesis testing," in \emph{Proc. IEEE International Symposium on Information Theory Proceedings}, July 2012, pp. 2191-2195.
\bibitem{Kartik_2019}
D. Kartik, A. Nayyar, and U. Mitra, ``Active hypothesis testing: beyond chernoff-stein," in \emph{Proc. IEEE International Symposium on Information Theory (ISIT)}, July 2019, pp. 897-901.
\bibitem{Coleman_2009}
T.P. Coleman, ``A stochastic control viewpoint on ‘posterior matching’-style feedback communication schemes," in \emph{IEEE International Symposium on Information Theory (ISIT)}, June 2009, pp. 1520-1524.
\bibitem{McEliece_1983}
R. J. McEliece, ``Communication in the presence of jamming-an information-theoretic approach," In \emph{Secure Digital Communications}, Springer, Vienna, pp. 127-166, 1983.
\bibitem{Borden_1985}
J. M. Borden, D. M. Mason, and R.J. McEliece, ``Some information theoretic saddlepoints," \emph{SIAM journal on control and optimization}, vol. 23, no. 1, pp. 129-143, 1985.
\bibitem{Stark_1988}
W. E. Stark, and R. J. McEliece, ``On the capacity of channels with block memory," \emph{IEEE transactions on information theory}, vol. 34, no. 2, pp. 322-324, 1988.
\bibitem{Palomar_2003}
D. P. Palomar, J. M. Cioffi, and M. A. Lagunas, ``Uniform power allocation in MIMO channels: A game-theoretic approach," \emph{IEEE Transactions on Information Theory}, vol. 49, no. 7, pp. 1707-1727, 2003.
\bibitem{Gupta_2008}
S. Gupta, K. P. Ramesh, and E. P. Blasch, ``Mutual information metric evaluation for pet/mri image fusion," in \emph{Proc. IEEE National Aerospace and Electronics Conference}, July 2008, pp. 305-311.
\bibitem{Friston_2017}
K. Friston, T. FitzGerald, F. Rigoli, P. Schwartenbeck, and G. Pezzulo, ``Active inference: a process theory," \emph{Neural computation}, vol. 29, no. 1, pp, 1-49, 2017.
\bibitem{Witsenhausen}
H. Witsenhausen, ``Separation of estimation and control for discrete time systems," \emph{Proceedings of the IEEE}, vol. 59, no. 11, pp. 1557-1566, 1971.
\bibitem{Krishnamurthy_2016}
V. Krishnamurthy, \emph{From optimal filtering to controlled sensing}, Cambridge University Press, 2016.
\end{thebibliography}
%                      Section Change                      %
%----------------------------------------------------------%

%\begin{IEEEbiography}{Konstantinos Ntemos}
%\end{IEEEbiography}

% \vfill
\vspace{-1.25cm}
\begin{IEEEbiographynophoto}{Konstantinos Ntemos}
received the BS degree in Computer Science in 2010, the MS degree in Administration and Economics of Telecommunications Networks in 2013 and the PhD degree in 2019 from the Dept. of Informatics and Telecommunications, National and Kapodistrian University of Athens, Greece. %He is currently PhD student in Department of Informatics and Telecommunication, University of Athens. 
His research interests include Stochastic Control, Game Theory and Machine Learning for multi-agent networks.
\end{IEEEbiographynophoto}
%\vfill
\vspace{-1.25cm}
\begin{IEEEbiographynophoto}{George Pikramenos}
received his MSci in Mathematics (First Class Honours) from Imperial College London, UK in 2016. He is currently a Phd student in the Dept. of Informatics and Telecommunications, National and Kapodistrian University of Athens, Greece. His research interests include Machine Learning, Private and Secure Computations and Stochastic Control.
\end{IEEEbiographynophoto}
%\vfill
\vspace{-1.25cm}
\begin{IEEEbiographynophoto}{Nicholas Kalouptsidis}
received the BS degree in mathematics %(with highest honors) 
from the University of Athens, in 1973 and the M.S and PhD degrees in systems science and mathematics from Washington University at St. Louis, MO, in 1975 and 1976, respectively. From 1989 until today he is professor of Communications and Signal Processing in the Department of Computer Science and communications, at the National and Kapodistrian University of Athens. He was a visiting scholar at Harvard University in 2008 and a visiting professor at Stanford University in 2015. He has more than 200 publications and 3 books.
\end{IEEEbiographynophoto}
%\vfill


%\vfill


\end{document}